\documentclass[a4paper, aps, pra, reprint, nofootinbib, superscriptaddress]{revtex4-1}
\usepackage{amsmath}
\usepackage{amssymb}
\usepackage{amsthm}
\usepackage{mathrsfs}
\usepackage{mathtools} 
\usepackage{verbatim}

\newcommand\pgen{p_{\mathrm{gen}}}
\newcommand\pswap{p_{\mathrm{swap}}}
\newcommand\pfuse{p_{\mathrm{fuse}}}

\usepackage{amsmath}  
\usepackage[ruled]{algorithm2e}
\usepackage[shortlabels]{enumitem}  
\usepackage{datetime}  
\usepackage[T1]{fontenc}

\newtheorem{theorem}{Theorem}
\newtheorem{remark}{Remark}

\newtheorem{df}{Definition}
\newtheorem{lemma}{Lemma}
\newtheorem{prop}{Proposition}
\newtheorem{example}{Example}
\newtheorem{cor}[theorem]{Corollary}
 					 
\usepackage{thm-restate}

\usepackage{physics}
\usepackage{hyperref}
\usepackage{graphicx}
\usepackage{subcaption}
\usepackage{amsmath}
\usepackage{amssymb}
\usepackage{tikz}
\usepackage{multirow}
\usepackage[cm]{fullpage}
\usepackage[T1]{fontenc}
\usepackage{todonotes}

\newcommand{\mean}[1]{E\left[#1\right]}

\def\Tupper{T^{\textnormal{upper}}}

\let\oldnl\nl
\newcommand{\nonl}{\renewcommand{\nl}{\let\nl\oldnl}}
\def\generate{\mbox{\sc generate}}

\def\restartuntilsuccess{\mbox{\sc restart-until-success}}
\def\swapuntilsuccess{\mbox{\sc swap-until-success}}
\def\distilluntilsuccess{\mbox{\sc distill-until-success}}

\def\Tgen{T_{\textnormal{gen}}}
\def\Tupper{T^{\textnormal{upper}}_{\textnormal{gen}}}
\def\Tapprox{T^{\textnormal{approx}}_{\textnormal{gen}}}
\def\Texp{T_{\textnormal{exp}}}

\def\stgeq{\geq_{\textnormal{st}}}
\def\stleq{\leq_{\textnormal{st}}}

\def\mugen{\mu_{\textnormal{gen}}}
\def\mugenupper{\mu_{\textnormal{gen}}^{\textnormal{upper}}}
\def\mupper{m_{\textnormal{upper}}}
\def\mlower{m_{\textnormal{lower}}}
\def\Tafter{T_{\textnormal{output}}}
\def\Tbefore{T_{\textnormal{input}}}
\def\Trepeater{\mathcal{T}}
\def\Tarm{T_{\textnormal{arm}}}
\def\Tswitch{T_{\textnormal{switch}}}

\begin{document}
\title{Improved analytical bounds on delivery times of long-distance entanglement}

\author{Tim Coopmans}
\email{t.j.coopmans@tudelft.nl}
\affiliation{QuTech, Delft University of Technology, The Netherlands}

\author{Sebastiaan Brand}
\email{s.o.brand@liacs.leidenuniv.nl}
\affiliation{Leiden Institute of Advanced Computer Science, Leiden University, The Netherlands.}

\author{David Elkouss}
\email{d.elkousscoronas@tudelft.nl}
\affiliation{QuTech, Delft University of Technology, The Netherlands}

\begin{abstract}
	The ability to distribute high-quality entanglement between remote parties is a necessary primitive for many quantum communication applications.
A large range of schemes for realizing the long-distance delivery of remote entanglement has been proposed, both for bipartite and multipartite entanglement.
For assessing the viability of these schemes, knowledge of the time at which entanglement is delivered is crucial.
Specifically, if the communication task requires multiple remote-entangled quantum states and these states are generated at different times by the scheme, the earlier states will need to wait and thus their quality will decrease while being stored in an (imperfect) memory.
For the remote-entanglement delivery schemes which are closest to experimental reach, this time assessment is challenging, as they consist of nondeterministic components
such as probabilistic entanglement swaps.
For many such protocols even the average time at which entanglement can be distributed is not known exactly, in particular when they consist of feedback loops and forced restarts.
In this work, we provide improved analytical bounds on the average and on the quantiles of the completion time of entanglement distribution protocols in the case that all network components have success probabilities lower bounded by a constant.
A canonical example of such a protocol is a nested quantum repeater scheme which consists of heralded entanglement generation and entanglement swaps.
For this scheme specifically, our results imply that a common approximation to the mean entanglement distribution time, the 3-over-2 formula, is in essence an upper bound to the real time.
Our results rely on a novel connection with reliability theory.

\end{abstract}
\maketitle

\section{Introduction}

The Quantum Internet is a vision of a world-wide network of nodes with the capability to transmit and process quantum information \cite{kimble2008quantum,wehner2018quantum}.
Such a network enables tasks that are impossible classically, among which unconditionally-secure communication \cite{bennett1984quantum, ekert1991quantum}, secure delegated computing \cite{childs2005secure} and extending the baseline of telescopes \cite{kellerer2014quantum}.
A primitive for such tasks is entanglement between remote nodes.
For establishing entanglement over distances beyond the fundamental distance limit~\cite{takeoka2014fundamental}, several schemes have been proposed, all making use of intermediate nodes~\cite{munro2015inside}.
These proposals include chains of quantum repeaters \cite{briegel1998quantum, munro2015inside, muralidharan2016optimal} and generalizations to two-dimensions for serving multiple users \cite{vardoyan2019stochastic, pirker2017modular, wallnoefer2019multipartite, pant2017routing, kuzmin2019scalable, wallnoefer2016twodimensional, das2018robust}.

Knowledge of the time that quantum repeater schemes take to deliver entanglement is highly relevant, for several reasons.
Most evidently, the entanglement should be delivered sufficiently fast for the application.
Secure communication over video, for example, requires transmission rates of at least hundreds of kbits per second \cite{schmidt2016mbps}.
Furthermore, for the repeater proposals which make use of quantum memories and do not rely on error correcting codes, i.e. the ones that are closest to experimental reach, the delivery time influences the quality of the produced entanglement.
The reason for this is that in these schemes, an entangled pair that is generated often needs to wait for another pair before the scheme can continue, and decoheres in memory while waiting.
In addition, some memory types suffer from effects which are effectively time-dependent.
For instance, noise on carbon spins in nitrogen-vacancy centres which is induced each time ones attempts to generate remote entanglement \cite{kalb2018dephasing}.
Another example is the decrease of the probability of extracting the state from an atomic-ensemble based quantum memory \cite{askarani2020frequency}.
Thus, the quality of the produced entanglement is a function of the time its generation takes.
This implies that knowledge of the delivery time is crucial for assessing the viability of schemes for long-distance entanglement distribution using near-term hardware.

Analysis of the delivery time is generally challenging for the entanglement-distribution schemes that are closest to experimental reach because they consist of probabilistic components.
The time such a scheme takes to deliver the entanglement, the completion time, is not a single number but instead a random variable.
For many schemes, the completion time is complex to express due to feedback loops and restarts.
Although numerically, progress has recently been made in determining the completion time for increasingly larger networks \cite{vanmeter2007system, shchukin2017waitingPRA, brand2020efficient, li2020efficient, kuzmin2019scalable, caleffi2017optimal}, numerical approaches provide only limited intuition and moreover are demanding in computation time when performing large-scale optimization over many network designs and hardware parameters.
For this reason, analytical results are more convenient.

Unfortunately, due to the complexity of the problem, even the average completion time is known exactly only in limited cases: for quantum repeater chains consisting of at most four repeater nodes \cite{shchukin2017waitingPRA, vinay2019statistical} and a star network with a single node in the center and an arbitrary number of leaf nodes \cite{vardoyan2019stochastic}.
For larger networks, analytical results only include approximations or loose bounds on the mean entanglement delivery time \cite{khatri2019practical}.
The approximations are based on the assumption that the success probabilities of some of the network components are very small \cite{sangouard2011quantum, kuzmin2020diagrammatic, schmidt2019memory-assisted, collins2007multiplexed} or close to 1 \cite{bernardes2011rate,praxmeyer2013reposition, khatri2019practical}.
Neither approximations are ideal, since some success probabilities can be boosted by techniques such as multiplexing, while others are bounded well below 1 for some setups\cite{calsamiglia2001maximum}.
Indeed, numerics have shown for some of the approximations that they become increasingly bad as the size of the network grows \cite{shchukin2017waitingPRA, brand2020efficient}.
Another scenario in which the completion time probability distribution is brought back to a known form includes the discarding of entanglement \cite{santra2018quantum, chakraborty2019distributed}.
See \cite{azuma2021tools} for a review of the completion time analysis for entanglement distribution schemes.

A canonical use case which has found particularly much application is a symmetric nested repeater scheme \cite{briegel1998quantum,duan2001long} where at each nesting level two entangled pairs of qubits, spanning an equal number of nodes, are connected.
Consequently, the entanglement span doubles at each nesting level.
For this scheme, it was empirically known \cite{jiang2007fast} that for small success probabilities of connecting the pairs, the average time to in-parallel create both required initial pairs at each nesting level is roughly $3/2$ times the average time for a single pair.
This results in an approximation to the average completion time of the repeater scheme which is known as the $3$-over-$2$ formula and has been frequently used since \cite{jiang2007fast, simon2007quantum, brask2008memory,sangouard2007long-distance,simon2007quantum,sangouard2008robust,brask2008memory,sangouard2009quantum,bernardes2011rate,sangouard2011quantum,abruzzo2013quantum,munro2015inside,boone2015entanglement,muralidharan2016optimal,asadi2018quantum,piparo2019quantum,asadi2020longARXIVV1,sharman2020quantum,wu2020nearterm,liorni2020quantum}.
Analytically finding the exact factor, for an arbitrary number of nesting levels and for any value of the success probabilities, has been an open problem for more than ten years \cite{sangouard2011quantum}.

In this work, we provide analytical bounds on the completion time which not only improve significantly upon existing bounds, but also show \textit{how good} some of the previous approximations are because the bounds become exact in the small probability limit.
To be precise, we give analytical bounds on the mean and quantiles of the completion time random variable for entanglement-distributing protocols which are constructed of probabilistic components whose success probability can be bounded by a constant from below.
This includes feedback loops in which failure of one component requires restart of other components, as long as no two components wait for the same other component to finish.
Regarding the symmetric nested repeater protocol, our bounds imply that the 3-over-2 approximation is, in essence, an upper bound to the mean completion time, rigorously rendering analyses based on this approximation pessimistic.
Other protocols we can treat include nested repeater chains with distillation and multipartite-entanglement generation schemes \cite{nickerson2012topological, vardoyan2019stochastic, kuzmin2019scalable}, among others.

This work is organized as follows.
First, in Sec.~\ref{sec:preliminaries} we describe the class of protocols our bounds apply to and introduce concepts from reliability theory we will use in the bounds' derivation.
Sec.~\ref{sec:results} contains our main results: analytical bounds on the mean completion time of such protocols and the tail of its probability distribution.
Next, we obtain improved bounds with respect to existing work by applying these results to two use cases: a nested quantum repeater chain (Sec.~\ref{sec:application-repeater}) and a quantum switch in a star network (Sec.~\ref{sec:application-switch}).
We finish with a discussion in Sec.~\ref{sec:discussion}.

\section{Preliminaries \label{sec:preliminaries}}

\subsection{Protocols}
\label{sec:preliminaries-protocols}

The protocols considered in this work aim to generate bipartite or multipartite entanglement between remote parties.
We will refer to bipartite entanglement as a `link'.
We consider protocols that are constructed from two building blocks: $\generate$ and $\restartuntilsuccess$.
Below, we explain the two building blocks individually, followed by describing how to build protocols from them.

\subsubsection{The $\generate$ building block}
First, by $\generate$ we refer to heralded generation of fresh entanglement, i.e. entanglement between remote nodes that is not produced from existing remote entanglement.
For simplicity, we will assume that the entanglement is bipartite and we will refer to such entanglement as an `elementary link'.
In our model, entanglement generation is performed in discrete attempts of fixed duration, each of which succeeds with a given constant probability $\pgen$ \cite{munro2015inside}.
The success is heralded, i.e. the nodes are aware which attempts fail and which succeed.
The duration of a single attempt equals $L / c$, where $L$ is the distance between the nodes and $c$ is the speed of light in the transmission medium.
We use $L / c$ as the unit of time.
As a consequence, the completion time of entanglement generation, denoted as $\Tgen$, is a discrete random variable following the geometric distribution:
\begin{equation}
	\label{eq:geometric-distribution}
\Pr( \Tgen = t) =
\begin{cases}
	\pgen (1 - \pgen)^{t-1} \mbox{ if $t\geq 1$ is an integer}\\
	0 \mbox{ otherwise.}\\
\end{cases}
.
\end{equation}
We will denote the mean of this distribution by $\mugen = 1 / \pgen$.

We will also consider the exponential distribution, which is the continuous analogue of the geometric distribution and is defined as follows: if $X$ follows the exponential distribution with parameter $\lambda > 0$, then
\begin{equation}
	\label{eq:exponential-distribution}
    \Pr( X > x) = e^{-\lambda x}
\end{equation}
for any real number $x \geq 0$.
For small $\pgen$, the completion time of entanglement generation is sometimes approximated by an exponential random variable $\Tapprox$ with the same mean, which is achieved by setting $\lambda = 1 / \mugen$.

\subsubsection{The $\restartuntilsuccess$ building block}
We introduce the next building block, $\restartuntilsuccess$, by example.
For this, we first describe two operations on existing entanglement: entanglement swapping and entanglement distillation.

By an entanglement swap~\cite{zukowski1993eventready} at node $M$, we refer to the operation which converts two links, one between nodes $A$ and $M$ and one between $M$ and $B$, into a single long-distance link between $A$ and $B$.
We model the entanglement swap as a probabilistic operation; in case the entanglement swap fails, both input links are lost.
By $\swapuntilsuccess$, we refer to the process which performs the following loop: it repeatedly takes two links $A-M$ and $M-B$ as input, followed by performing an entanglement swap on them, while the process only terminates if the entanglement swap was successful.
That is, if the swap failed, then the protocol requires the input links to be regenerated.
This process repeats until the entanglement swap succeeds.
We explicitly do not specify how the input links were produced.
These could each be for example delivered by the $\generate$ block, but they could for instance also the result of a succesful entanglement swap themselves.

We assume that the swap success probability $0 < \pswap \leq 1$ is a constant that is independent of the states upon which the swap acts.
This assumption is valid when the input states to the entanglement swap are Bell-diagonal, i.e. probabilistic mixtures of the four Bell states
\[
    \ket{\Phi^{\pm}} = \frac{\ket{00} \pm \ket{11}}{\sqrt{2}},
    \quad
    \ket{\Psi^{\pm}} = \frac{\ket{01} \pm \ket{10}}{\sqrt{2}}
    .
\]
Such a scenario arises, for example, when all imperfections are modelled as the random application of single-qubit Pauli gates~\cite{nielsen2000quantum}, because these permute the four Bell states.
In particular, each Bell state can be mapped to a single target Bell state, say $\ket{\Psi^+}$, by applying a single-qubit Pauli operator to each of the qubits that remain at node $A$ and $B$.
Since only the qubits at node $M$ are involved in the operation that performs the entanglement swap, the success probability of an entanglement swap for any of the 16 combinations of input states is identical to the success probability in case both input states were $\ket{\Psi^+}$.
We thus see that the success probability in case of Bell-diagonal states is a constant, independently of which scheme is used for performing the entanglement swap.

We model fusion, the generalization of the entanglement swap which converts more than 2 input links to a multipartite entangled state, in similar fashion to the entanglement swap.

Entanglement distillation is the probabilistic conversion of two low-quality links shared between two nodes to a single high-quality link between the same two nodes \cite{bennett1996purification, deutsch1996quantum}.
The success probability of distillation depends on the states of the two links, and is lower bounded by $\frac{1}{2}$ for the schemes considered here.
Similarly to the case of entanglement swapping, the two input links are lost if the distillation step fails.
By $\distilluntilsuccess$ we denote the analog of $\swapuntilsuccess$ where the probabilistic operation is entanglement distillation.

We assume that the durations of the entanglement swap, fusion, and distillation operations are negligible.

In general, we use the term $\restartuntilsuccess$ for an operation which takes entanglement as input, performs a probabilistic operation onto it, and demands the regeneration of the input entanglement in the case of failure.
Its success probability can be a function of properties of the input entanglement, such as its quality or its delivery time, but it may also be a constant.
Thus, $\swapuntilsuccess$ and $\distilluntilsuccess$ are instantiations of $\restartuntilsuccess$ where the probabilistic operation is entanglement swapping and entanglement distillation, respectively.
For clarity, we emphasize that for any $\restartuntilsuccess$ operation, all input entanglement needs to be present before the operation can be performed.

\subsubsection{Building protocols from the two building blocks}
\label{sec:building-protocols-from-building-blocks}
The protocols we consider in this work are composed from heralded entanglement generation and $\restartuntilsuccess$ as subprotocols, with the restriction that the distinct $\restartuntilsuccess$ protocols do not compete for the same resources.
That is, no pair of subprotocols waits for the same link before proceeding.
This corresponds to the protocols where the dependency graph of the inputs and outputs of the subprotocols is a tree.
Consequently, the order in which the various probabilistic operations (such as entanglement swaps) are performed, is fixed.
Fig.~\ref{fig:tree} visualizes this tree structure by showing examples of such protocols (see figure caption for further explanation).

As a concrete example, consider the $\distilluntilsuccess$ protocol on two nodes in fig.~\ref{fig:tree}(b).
The protocol starts with Alice and Bob generating two links in parallel using heralded entanglement generation ($\generate$).
When both links are ready, they perform entanglement distillation, which is a probabilistic operation.
If distillation fails, the two input links are lost.
Consequently, Alice and Bob perform heralded entanglement generation again, after which they attempt entanglement distillation once more.
This procedure is repeated until the distillation operation succeeds.
This example protocol is a specific instance of a $\restartuntilsuccess$ protocol because the protocol (i.e. the sequence: $\generate$ twice in parallel, followed by distillation) is restarted when entanglement distillation fails.
Moreover, it can be used as a subprotocol when, for example, the link it outputs is used as a (partial) input to another operation, such as an entanglement swap (see fig.~\ref{fig:tree}(c) for an example).

Due to the probabilistic nature of $\generate$ and of the restarts, the completion time of a $\restartuntilsuccess$ protocol is a random variable.
Since we defined $\Tgen$, the completion time of $\generate$, as a discrete random variable, so is the completion time of any $\restartuntilsuccess$ protocol in which elementary links are produced using $\generate$.
However, at the start of Sec.~\ref{sec:results}, we will consider a continuous random variable as alternative to $\Tgen$.
In that case, the completion time of $\restartuntilsuccess$ will also be continuous.

\begin{figure*}
    \centering
    \includegraphics[width=1.0\textwidth]{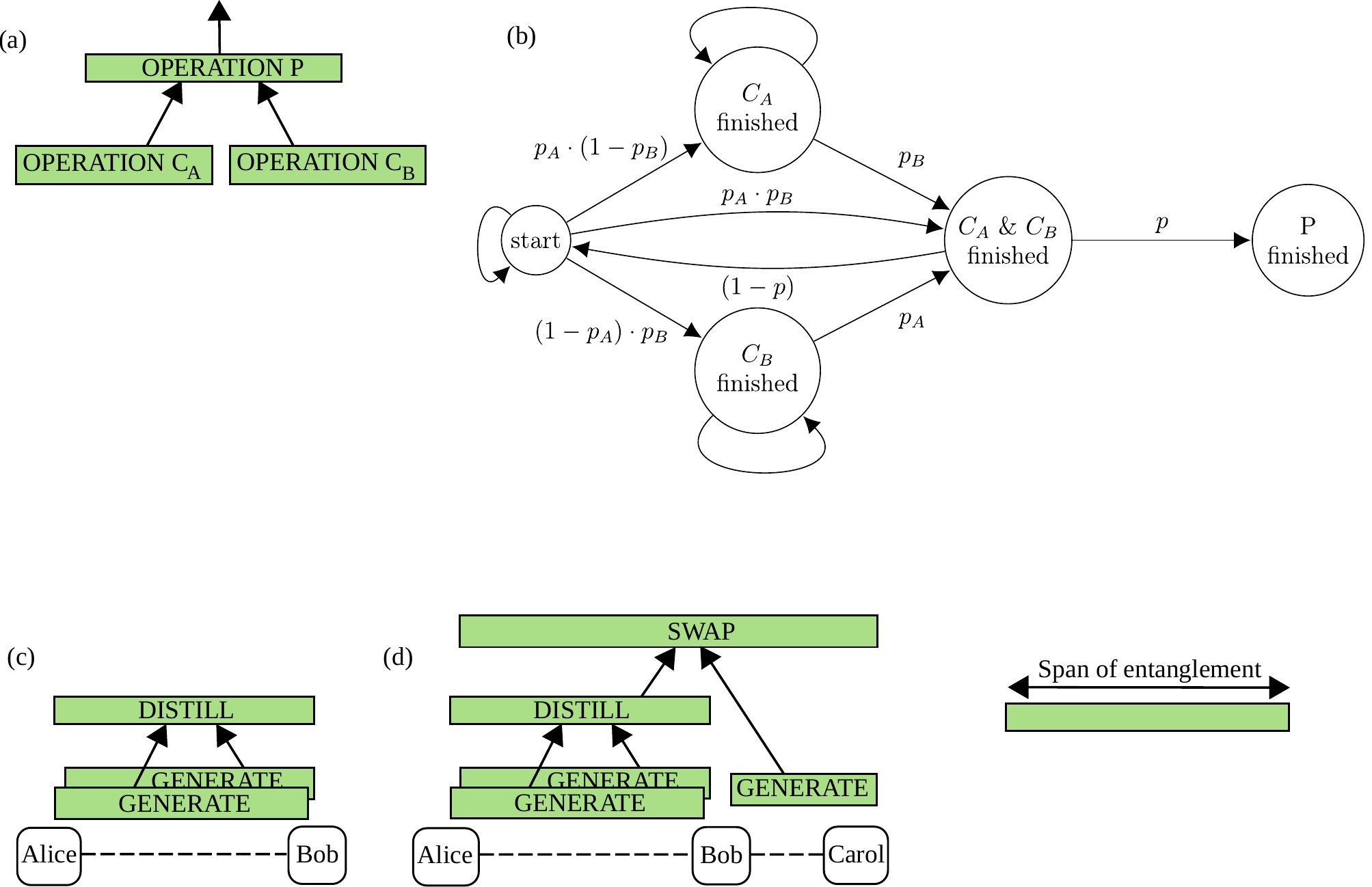}
	\caption{
        The results in this work bound the completion time of any entanglement-distribution protocol which can be visualized as a tree.
        \textbf{(a)}
        In such a tree, each vertex is labelled by the operation ({\sc P}) that should be performed as soon as the operations on the vertex's children ({\sc $\textnormal{C}_{\textnormal{A}}$} and {\sc $\textnormal{C}_{\textnormal{B}}$}) have finished.
        In case the operation fails, both children start regenerating entanglement, possibly by recursively having their children regenerate entanglement.
        This procedure is repeated until the operation {\sc P} succeeds.
        \textbf{(b)} the different states of the tree in (a), and the corresponding transition probabilities (unlabeled transition probabilities left implicit since a node's transition probabilities sum to $1$).
        Here $p_A / p_B / p$ are the success probabilities of the operations $C_A / C_B / P$, which in general need not be constant but may for example depend on the quantum states the the operation acts upon, such as in entanglement distillation.
        We emphasize that ${\sc P}$ will only be attempted once {\sc $\textnormal{C}_{\textnormal{A}}$} and {\sc $\textnormal{C}_{\textnormal{B}}$} have finished, and that moreover if ${\sc P}$ fails, then the process is restarted.
        \textbf{(c)} Example protocol on two nodes, Alice and Bob, which consists of performing heralded entanglement generation ({\sc generate}) twice in parallel, followed by entanglement distillation ({\sc distill}) on the two freshly generated links.
        In case of failure of the distillation attempt, both links are lost, in which case the protocol restarts.
        This procedure is repeated until the distillation attempt succeeds.
        \textbf{(d)} Example protocol on three nodes.
        Alice and Bob perform the protocol from (b), and in parallel Bob and Carol perform heralded entanglement generation.
        As soon as both have finished, Bob performs an entanglement swap
 ({\sc swap}).
        This procedure is repeated until the swap succeeds.
	\label{fig:tree}
	}
\end{figure*}

\subsection{Probability theory and the NBU property}

In this work, we will make extensive use of a class of probability distributions called new-better-than-used (NBU), which have been studied in the context of reliability theory and life distributions~\cite{marshall2007life}.
In order to mathematically define new-better-than-used, we first revisit some notions from probability theory.
All random variables in this work that are continuous have the positive reals as domain, i.e. a continuous random variable $X$ with $\Pr(X < 0) = 0$.
The cumulative distribution function (CDF) of random variable $X$ is $x \mapsto \Pr(X \leq x)$, and the co-CDF is $x \mapsto \Pr(X > x)$.
This co-CDF is also referred to as the survival function or the \textit{reliability}, since it states the probability that $X$ will survive at least up to time $x$.
The residual life distribution of $X$ is given by the conditional probability $\Pr(X > x + y | X > y)$ and describes the time that $X$ will survive at least up another interval $x$ given that it has already survived time $y$.
We now say that a real-valued random variable $X$ is new-better-than-used (NBU) or that it has the NBU property if its residual life distribution is upper bounded by the original reliability, i.e.
\begin{equation}
	\label{eq:nbu-equivalent-def}
	\forall x, y \geq 0: \qquad \Pr(X > x + y | X > y) \leq \Pr(X > x)
	.
\end{equation}
Intuitively, new-better-than-used random variables describe ageing over time. 
As an example, consider the lifetime of a car: the probability that an old car (one that is already $y$ years old) will survive another $x$ years is smaller than the probability that a brand new car will reach the age of $x$ years.

For clarity, we separately state the definition of NBU, where we use an expression equivalent to eq.~\eqref{eq:nbu-equivalent-def} for convenience of our proofs later on.

\begin{df}
	\label{def:nbu}
	A real-valued random variable $X$ with $\Pr(X < 0) = 0$ , is called new-better-than-used (NBU) if
	\[
		\forall x,y\geq 0: \qquad	\Pr(X > x + y) \leq \Pr(X > x) \cdot \Pr(X > y)
		.
		\]
	It is called new-worse-than-used (NWU) if the reverse inequality holds.
\end{df}

We give two examples of NBU distributions.

\begin{example}
	A delta-peak distribution $\Pr(X = x_0) = 1$ for some fixed $x_0 \geq 0$ is NBU, since 
	\begin{eqnarray*}	
	\Pr(X > x) \cdot \Pr(X > y) = 
	  \begin{cases}
		  1  & \text{if } x< x_0 \textnormal{ and } y < x_0\\
		  0 & \text{otherwise}
	  \end{cases}
	\end{eqnarray*}
	while
	\begin{eqnarray*}	
	\Pr(X > x + y) =
	  \begin{cases}
		  1  & \text{if } x + y < x_0\\
		  0 & \text{otherwise.}
	  \end{cases}
	\end{eqnarray*}
	Since $x + y < x_0$ implies $x < x_0$ and $y < x_0$ for any $x,y \geq 0$, we see that $\Pr(X > x + y) \leq \Pr(X > x) \cdot \Pr(X > y)$ and thus $X$ is NBU.
\end{example}

\begin{example}
	\label{example:exponential-distribution}
    The exponential distribution, defined in eq.~\eqref{eq:exponential-distribution}, satisfies $\Pr(X > x + y ) = \Pr(X > x) \cdot \Pr(X > y)$ for all $x, y \geq 0$ and is therefore both NBU and NWU.
\end{example}

Lastly, we will use the notion of stochastic dominance.

\begin{df}
    \label{def:stochastic-dominance}
    Let $X$ and $Y$ be two random variables with common domain $D$, a subset of the real numbers.
    We say that $X$ stochastically dominates $Y$ and write $X \stgeq Y$ if
    \[
    \Pr(X > z) \geq \Pr(Y > z)
    \]
    for all $z \in D$.
\end{df}

In particular, we will use the following lemma, which states that stochastic dominance of one random variable over the other implies an ordering of their means.

\begin{lemma}
    \label{lemma:stochastic-dominance-means}
    Let $X$ and $Y$ be two random variables with domain $[0, \infty)$.
If $X \stgeq Y$, then $E[X] \geq E[Y]$.
\end{lemma}
\begin{proof}
    The lemma directly follows from the definition of stochastic dominance, together with the fact that the mean of $X$ can be written as an integral over the co-CDF,
    \[ E[X] = \int_{0}^{\infty} \Pr(X > x) dx
    ,
    \]
    and similarly for $Y$.
\end{proof}

\section{Main results\label{sec:results}}

In this section, we give our main results in Prop.~\ref{prop:swap-bounds} and \ref{prop:swap-bounds-non-iid}: bounds on the completion time distribution for protocols composed of elementary-link generation ($\generate$) and $\restartuntilsuccess$ operations.
The proofs to the main results can be found in Sec.~\ref{sec:methods}.

Our results bound continuous completion times, whereas the completion time of elementary-link generation is the discrete random variable $\Tgen$ (see Sec.~\ref{sec:preliminaries}).
Therefore, before stating our main results we first remark that $\Tgen$ is stochastically dominated by a continuous NBU random variable we denote as $\Tupper$.

\begin{lemma}
	\label{lemma:generation-majorization}
    The completion time $\Tgen$ of elementary-link generation is stochastically dominated (Def.~\ref{def:stochastic-dominance}) by the continuous random variable $\Tupper = 1+\Texp$ where $\Texp$ is exponentially distributed with parameter $\frac{-1}{\log(1 - \pgen)}$. That is,
    \begin{eqnarray*}
        \Pr(\Tgen > t) &\leq& \Pr(\Tupper > t) \\
        &=&
        \begin{cases}
            1 & \textnormal{if } 0 \leq t \leq 1 \\
            \exp\left(
            (t-1)/\log(1 - \pgen)
            \right)
            & \textnormal{if } t \geq 1
        \end{cases}
    \end{eqnarray*}
    The mean of $\Tgen$ is upper bounded by the mean of $\Tupper$ which is given by
    \begin{equation}
        \label{eq:mugenupper}
            \mugenupper = 1- \frac{1}{\log(1 - \pgen)} = \frac{1}{\pgen} + \frac{1}{2} + O(\pgen)
     \end{equation}
    where $O(\pgen)$ contains terms that scale with $\pgen$ or powers of it.
    The means of $\Tgen$ and $\Tupper$ differ only slightly, both in difference and in ratio:
    \begin{equation}
        \label{eq:mu-difference-bound}
        0 \leq \mugenupper - \mugen \leq \frac{1}{2}
        \text{ and }
        1 \leq \frac{\mugenupper}{\mugen} \leq 1 + \frac{\pgen}{2}
    \end{equation}
        for any $\pgen\in [0, 1]$.
		Moreover, $\Tupper$ is NBU.
\end{lemma}

\begin{figure}
\centering

\begin{subfigure}{0.5\textwidth}
\centering
    \includegraphics[width=\columnwidth]{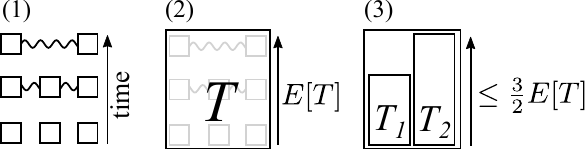}
    \caption{
    \label{fig:results-mean}
}
\end{subfigure}

\vspace{4mm}

\begin{subfigure}{0.5\textwidth}
\centering
    \includegraphics[width=0.70\columnwidth]{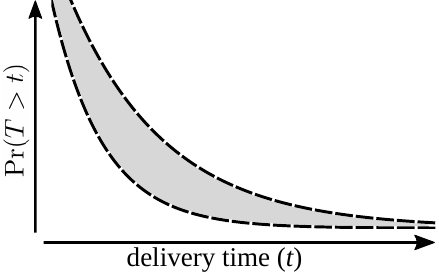}
    \caption{
    \label{fig:results-tail}
}
\end{subfigure}

\caption{\label{fig:results}
    Visual overview of this work's bounds on the completion time of entanglement distribution protocols.
    \textbf{(Top)} The first result is a bound on the mean completion time of two parallel entanglement distribution processes.
    To be precise, we consider an entanglement distribution process (1), whose completion time is a random variable $T$ and has mean $E[T]$ (2).
    If $T$ is NBU (Def.~\ref{def:nbu}), we show that completing two such independent and identically distributed processes in parallel has a mean time which is bounded from above by $\frac{3}{2}\cdot E[T]$ (3).
    \textbf{(Bottom)} Our second result is a two-sided bound on the probability distribution of the completion time of such processes.
    These bounds decay exponentially fast.
    }
\end{figure}

As consequence of Lemma~\ref{lemma:generation-majorization}, we may assume that the duration of elementary-link generation is described by $\Tupper$ if we are looking for upper bounds on a protocol's completion time.
Indeed, an upper bound on the co-CDF or the mean of the resulting completion time will automatically also become an upper bound on the real completion time (see Def.~\ref{def:stochastic-dominance} and Lemma~\ref{lemma:stochastic-dominance-means}).

Now let us state our bounds on continuous completion times.
For legibility, we first state a special case of our main result: the scenario where a $\swapuntilsuccess$ operation with constant success probability is performed on two quantum states.
We assume that the time it takes until a state is produced is a random variable, and that this random variable is the same for both input states; that is, their completion times are independent and identically distributed.

\begin{center}
\fbox{
\begin{minipage}{0.4\textwidth}
	\begin{center}
		\textbf{Completion time of swapping: two states \& IID}\\
	\end{center}
\begin{prop}

	\label{prop:swap-bounds}
	Consider the time $\Tafter$ of a $\swapuntilsuccess$ protocol with constant success probability $p$, acting on two quantum states, produced with identically-distributed independent completion times $\Tbefore$.
	If $\Tbefore$ is a continuous random variable and it is NBU (Def.~\ref{def:nbu}), then:
	\begin{enumerate}[(a)]
		\item $\Tafter$ is NBU;
		\item the mean of $\Tafter$ is upper bounded as 
            \[E[\Tafter] \leq  \frac{3E[\Tbefore]}{2p} ;\]
		\item for all $t$, the probability that $\Tafter$ takes longer than $t$ timesteps decays exponentially fast:
			\[\Pr(\Tafter > t) \leq  \exp\left(p - \frac{2 p t}{3 E[\Tbefore]}\right)\]
			while it is lower bounded as
			\[\Pr(\Tafter > t) \geq \exp\left(\frac{-2 p t}{3E[\Tbefore]} \cdot \frac{1}{1-p}\right).\]
        \item in the limit $p\rightarrow 0$, the normalized completion time $\Tafter / \mean{\Tafter}$ approaches the exponential distribution with mean 1, and thus $E[\Tafter] \cdot 2p/(3E[\Tbefore]) \rightarrow 1$.
	\end{enumerate}
\end{prop}
\end{minipage}
}
	\vspace{\baselineskip}
\end{center}

The bounds from Prop.~\ref{prop:swap-bounds} are visually depicted in Fig.~\ref{fig:results}.

Although Prop.~\ref{prop:swap-bounds} regards a $\swapuntilsuccess$ protocol, it also finds application to $\distilluntilsuccess$, which has nonconstant success probability:

\begin{remark}
	\label{remark:distillation}
	Consider Prop.~\ref{prop:swap-bounds} where $\swapuntilsuccess$ is replaced by $\distilluntilsuccess$.
	Note:
	\begin{enumerate}[(a)]
		\item{
			Prop~\ref{prop:swap-bounds}(a)-(c) still hold in case the quantum states produced with completion times $\Tbefore$ do not decohere over time, because then the distillation success probability $p$ is a constant, independent of the production times of the input states;
				}
	\end{enumerate}
				The success probability of distillation is general lower bounded by $1/2$, resulting in
    \begin{enumerate}[(a)]
            \setcounter{enumi}{1}
		\item{ 
			$E[\Tafter] \leq 3 E[\Tbefore]$.
				}
    \end{enumerate}
    Since the upper bound in Prop~\ref{prop:swap-bounds}(c) is monotonically decreasing in $p$ in the regime $t \geq 3E[\Tbefore] / 2$, we may replace $p$ by its lower bound $1/2$ to obtain:
    \begin{enumerate}[(a)]
            \setcounter{enumi}{2}
            \item{for $t \geq 3E[\Tbefore] / 2$, we have
                \begin{equation*}
                    \Pr(\Tafter > t) \leq \exp(\frac{1}{2}-\frac{t}{3E[\Tbefore]})
                    .
                \end{equation*}
				}
	\end{enumerate}
\end{remark}

Prop.~\ref{prop:swap-bounds} is a special case of a more general version of Prop.~\ref{prop:swap-bounds-non-iid} for $\restartuntilsuccess$ protocols that act on two or more quantum states whose completion times are independent, but not necessarily identically distributed.

\begin{center}
\fbox{
\begin{minipage}{0.4\textwidth}
	\begin{center}
        \textbf{General case: completion time of $\restartuntilsuccess$ protocol}\\
	\end{center}
\begin{prop}

	\label{prop:swap-bounds-non-iid}
	Consider the time $\Tafter$ of a $\restartuntilsuccess$ protocol with constant success probability $p$, acting on $n\geq 2$ quantum states, produced with independent completion times $T_1, \dots, T_n$, which need not be identically distributed.
	Suppose that each of $\Tafter$ and $T_1, \dots, T_n$ is a continuous random variable.
	Denote \mbox{$m = E[\max(T_1, \dots, T_n)]$}.
	If all $T_1, \dots, T_n$ are NBU (Def.~\ref{def:nbu}), then:
	\begin{enumerate}[(a)]
		\item $\Tafter$ is NBU;
        \item the mean of $\Tafter$ equals \mbox{$E[\Tafter] = m / p$};
		\item for all $t$, the probability that $\Tafter$ takes longer than $t$ timesteps is exponentially bounded from above as
			\[
			\Pr(\Tafter > t)
			\leq \exp\left(p - \frac{p \cdot t}{m}\right)
			.\]
			while it is bounded from below by
			\[
			\Pr(\Tafter > t)
			\geq
				\exp\left(\frac{-p\cdot t}{m} \cdot \frac{1}{1 - p}\right)
				.
				\]
		\item in the limit $p\rightarrow 0$, the normalized completion time $\Tafter / \mean{\Tafter}$ approaches the exponential distribution with mean 1, and thus $E[\Tafter] \cdot p/m \rightarrow 1$.
	\end{enumerate}
	\begin{enumerate}[(e)]
		\item{
			For general $n$, we have
	\[
		\max_{1\leq j \leq n} E[T_j] \leq m \leq \sum_{j=1}^n E[T_j]
            .
		\]
    }
	\end{enumerate}
	\begin{enumerate}[(f)]
        \item{
            If $n=2$, then we also have
            \begin{eqnarray*}
                m \leq &&\frac{3}{4} \cdot \left(E[T_1] + E[T_2]\right) \\
                &&+ \int_{0}^{\infty} \left[ \Pr(T_1 > t) - \Pr(T_2 > t)\right]^2 dt
                .
            \end{eqnarray*}
			}
	\end{enumerate}
	\begin{enumerate}[(g)]
		\item{
                In case all $T_j$ are identically distributed with mean $E[T]$, then a tighter bound than (e) exists:
	\[
		1 \leq
			\frac{m}{E[T]} \leq n - 1 + \frac{1}{n}
		.
		\]
			}
	\end{enumerate}
\end{prop}
\end{minipage}
}
	\vspace{\baselineskip}
\end{center}

We finish this section by generalizing Remark~\ref{remark:distillation}.

\begin{remark}
    \label{remark:bounded-success-probability}
Consider a {\restartuntilsuccess} protocol whose success probability is lower bounded by a constant $c$.
Then the upper bounds in Prop.~\ref{prop:swap-bounds-non-iid}(e-g) still hold, while 
    Prop.~\ref{prop:swap-bounds-non-iid}(b) and (c) can respectively be replaced by \mbox{$E[\Tafter] \leq m/c$} and \mbox{$\Pr(\Tafter > t) \leq \exp(c-\frac{ct}{m})$} for \mbox{$t \geq m$}.
\end{remark}

In the next sections, we give two use cases for the bounds derived in this section: a quantum repeater chain scheme and a quantum switch protocol.

\section{First application: nested quantum repeater chain \label{sec:application-repeater}}

In this section, we apply our bounds on the completion time of entanglement distribution protocols to an extensively-studied nested repeater chain protocol \cite{briegel1998quantum, duan2001long}.
We explain the protocol for the case where the number of segments is $2^n$ for some integer $n \geq 0$ (i.e. the chain consists of $2^{n}+1$ nodes).
See also Fig.~\ref{fig:schematic-nested-repeater}.
If $n=0$, then the network consists of two end nodes only (no repeaters), which use heralded entanglement generation (see Sec.~\ref{sec:preliminaries}) to generate a single elementary link.
If $n>0$, then the chain has a middle node (since the number of segments is even).
In parallel, a $2^{n-1}$-hop-spanning link is produced on the left side of the middle node, as well as a link on its right side.
As soon as both links have been prepared, the middle node performs an entanglement swap to convert the two links into a single $2^n$-hop-spanning link.
This scheme can also be extended with one or multiple rounds of entanglement distillation at each nesting level, in a nested fashion \cite{briegel1998quantum}.

The exact completion time distribution of the nested repeater scheme has so far not been analytically found beyond the single-repeater case.
The problem was first fully explained by Sangouard et al. \cite{sangouard2011quantum}, although it was already partially described in earlier work~\cite{jiang2007fast, simon2007quantum, brask2008memory}.
Sangouard et al., remarked that while the completion time of elementary-link generation at the bottom level follows a well-known distribution (the geometric distribution, Sec.~\ref{sec:preliminaries}), this is no longer the case for higher levels.

To circumvent this issue, many have resorted to approximating the probability distribution at each level with an exponential distribution, combined with the small-probability assumptions $\pswap \ll 1$ and $\pgen \ll 1$.
This approximation leads to an expression for the mean entanglement delivery time as follows.  At each nesting level, the protocol can only continue if both input states to the entanglement swap have been produced.
Mathematically, this is expressed as the maximum of the delivery time of the two links.
The mean of the maximum of two independent and identically distributed (i.i.d.) exponential random variables with mean $\mu$ is $\frac{3}{2} \cdot \mu$.
Next, if the swap success probability is $\pswap$, then on average $1/\pswap$ attempts are needed until success.
Thus, for each nesting level, the mean entanglement delivery time should be multiplied by a factor $3/(2\pswap)$, resulting into an expression for the mean delivery time known as the \textit{3-over-2-approximation}:
\begin{equation}
	\label{eq:3-over-2}
	\left(\frac{3}{2\pswap}\right)^n \cdot \frac{1}{\pgen}
    .
\end{equation}

The 3-over-2 approximation was first used by Jiang et al.\cite{jiang2007fast}, who mentioned that the factor $3/2$ agreed well with simulations in the small-probability regime.
Since then, the approximation has been frequently used \cite{sangouard2007long-distance,simon2007quantum,sangouard2008robust,brask2008memory,sangouard2009quantum,bernardes2011rate,sangouard2011quantum,abruzzo2013quantum,munro2015inside,boone2015entanglement,muralidharan2016optimal,asadi2018quantum,piparo2019quantum,asadi2020longARXIVV1,sharman2020quantum,wu2020nearterm,liorni2020quantum}.

In addition to the 3-over-2 approximation, Sangouard et al. \cite{sangouard2011quantum} noted that the repeater scheme's mean completion time can be bounded using the following remark: the mean of the maximum of two nonnegative i.i.d random variables with mean $\mu$ is lower bounded by $\mu$ and upper bounded by $2\mu$.
These bounds correspond to the scenario where one waits only for a single link to be ready, or for both links to be prepared sequentially, respectively.
Consequently,
\begin{equation}
    \label{eq:markov-repeater-mean-bound}
    \left(\frac{1}{\pswap}\right)^n \cdot \frac{1}{\pgen} 
    \leq E[T] \leq
\left(\frac{2}{\pswap}\right)^n \cdot \frac{1}{\pgen} 
.
\end{equation}
Now we use Markov's inequality, 
    $\Pr(T \geq t) \leq E[T]/t$, which can be rephrased
\begin{equation}
    \label{eq:markov-inequality}
    \Pr(T > t) \leq E[T] \cdot \frac{1}{t + 1},
\end{equation}
since $T$ only takes integral values.
Substituting $E[T]$ by its upper bound from eq.~\eqref{eq:markov-repeater-mean-bound} leads to
\begin{equation}
    \label{eq:markov-repeater-cdf-bound}
    \Pr(T > t) \leq \left(\frac{2}{\pswap}\right)^n \cdot \frac{1}{\pgen} \cdot \frac{1}{t + 1}
    .
\end{equation}
Both the mean bound from eq.~\eqref{eq:markov-repeater-mean-bound} and the tail bound from eq.~\eqref{eq:markov-repeater-cdf-bound} are quite loose bounds, see Fig.~\ref{fig:mean_bound_ratios_plot} and \ref{fig:tail_bound_plot}.
Only recently, it was shown analytically by Kuzmin and Vasilyev that the factor 3/2 from eq.~\eqref{eq:3-over-2} is exact in the limit of vanishing swap success probability, and moreover that the delivery time probability distribution after an entanglement swap in this limit is indeed an exponential distribution \cite{kuzmin2020diagrammatic}.

Our bounds from Sec.~\ref{sec:results} allow us to go beyond these results.
In particular, we show the following.
First, we analytically show that the 3-over-2 approximation is, in essence, an \textit{upper bound} to the mean completion time.
This implies that the 3-over-2 approximation is pessimistic, confirming numerical simulations \cite{bernardes2011rate, shchukin2017waitingPRA}.
Next, we derive two-sided bounds on the tail of the probability distribution of the repeater chain's completion time.
Both the mean bound and the tail bounds coincide in the limit of vanishing success probabilities.
We give the bounds below and plot them in Fig.~\ref{fig:mean_bound_ratios_plot} (mean bounds) and Fig.~\ref{fig:tail_bound_plot} (tail bounds).

\begin{prop}
	\label{prop:repeater-bounds}
    Consider the completion time $\Trepeater_n$ of an equally-spaced, symmetric nested repeater scheme (no distillation) on $2^n$ segments, such as the example in Fig.~\ref{fig:schematic-nested-repeater} for $n=2$. If $n>0$, then:
	\begin{enumerate}[(a)]
        \item \label{item:repeater-chain-mean-bounds}
        the mean completion time is upper bounded as
	\begin{eqnarray*}
			E[\Trepeater_n]
			&\leq&
			\left(\frac{3}{2\pswap}\right)^n \cdot \mu_0
            .
	\end{eqnarray*}
    Here, $\mu_0$ is the mean of any real-valued NBU random variable which stochastically dominates the completion time $\Tgen$ of elementary-link generation.
	In case the elementary-link generation is modelled as discrete attempts which succeed with probability $\pgen$, then we choose $\Tupper$ for this random variable (see Lemma~\ref{lemma:generation-majorization}), resulting in
			\[
                \mu_0 = E[\Tupper] = 1 - \frac{1}{\log(1-\pgen)}
                .
			\]
            If instead the completion time of elementary-link generation is described by the exponentially-distributed random variable $\Tapprox$ (see Sec.~\ref{sec:preliminaries-protocols}), which is NBU itself, then \mbox{$\mu_0 = E[\Tapprox] = 1/\pgen$}.
            By Lemma~\ref{lemma:generation-majorization}, the two models' means only differ slightly: \mbox{$0\leq E[\Tupper] - E[\Tapprox]\leq\frac{1}{2}$} and \mbox{$1 \leq E[\Tupper] / E[\Tapprox] \leq 1 + \pgen / 2$}.
        \item the mean completion time is lower bounded as
            \[
				E[\Trepeater_n]
				\geq
                \frac{1}{\pswap} \cdot
            \left(\frac{3 - 2\pswap}{\pswap (2 - \pswap)}\right)^{n-1} \cdot \nu_0
				.
				\]
                Here, $\nu_0$ is the mean time until the latest of two parallel elementary-link generation processes has finished.
                In case elementary-link generation is modelled as discrete attempts which succeed with probability $\pgen$, then 
            \[
                \nu_0 = \frac{3 - 2\pgen}{\pgen(2 - \pgen)}
            \]
            while if its completion time is modelled by an exponential distribution, then $\nu_0 = 3/(2\pgen)$.
	\item \label{item:repeater-chain-tail-bounds}
        the co-CDF of $\Trepeater_n$ differs from the co-CDF of an exponential distribution by at most a factor $\exp(\pswap)$ from above,
		\[
            \Pr(\Trepeater_n > t) \leq \exp(\pswap)\cdot \exp(- \frac{\pswap \cdot t}{\mupper})
			\]
            while it is lower bounded as
		\[
            \Pr(\Trepeater_n > t) \geq \exp(\frac{-\pswap \cdot t}{\mlower} \cdot \frac{1}{1 - \pswap})
            .
			\]
           Here, we have denoted 
            \[ 
            \mupper = \frac{3}{2} \cdot \left(\frac{3}{2\pswap}\right)^{n-1} \cdot \mu_0
            \]
            and
            \[
                \mlower = \left(\frac{3 - 2\pswap}{\pswap (2 - \pswap)}\right)^{n-1} \cdot \nu_0
            \]
            where $\mu_0$ and $\nu_0$ are given in Prop.~\ref{prop:repeater-bounds}\ref{item:repeater-chain-mean-bounds} and (b).
        \item in the limit where both $\pswap \rightarrow 0$ and $\pgen \rightarrow 0$, the normalized random variable $\Trepeater_n/E[\Trepeater_n]$ follows the exponential distribution with mean 1, and moreover
            \[
                \lim_{\pswap\rightarrow 0, \pgen \rightarrow 0} E[\Trepeater_n]/L_n = 1
            \]
            with 
            \[
                L_n = \left(\frac{3}{2\pswap}\right)^{n} \cdot \frac{1}{\pgen}
                .
            \]
        \item If the completion time of elementary-link generation is described by the exponentially-distributed $\Tapprox$, then $\Trepeater_n$ is NBU, while if it is modelled as discrete attempts, then $\Trepeater_n$ is stochastically dominated (Def.~\ref{def:stochastic-dominance}) by an NBU random variable which satisfies the bounds in items (a-c).
	\end{enumerate}
\end{prop}

Most statements in Prop.~\ref{prop:repeater-bounds} directly follow by applying Prop.~\ref{prop:swap-bounds} in Sec.~\ref{sec:results} iteratively over the number of nesting levels.
In particular, a useful feature following from Prop.~\ref{prop:swap-bounds}(a) is that at each nesting level, the completion time possesses the NBU property (Def.~\ref{def:nbu}).
Consequently, the mean upper bound in Prop.~\ref{prop:swap-bounds}(c), which is only applicable to NBU random variables, can be used at each nesting level.
Only the lower bound in (b) and the expression for $\mlower$ in (c) do not follow from Prop.~\ref{prop:swap-bounds}.
These can be found by noting that the maximum of two sums dominates a single sum whose length is the maximum of the original two sum lengths.
We give the full proof in Appendix~\ref{app:repeater-bounds}.

\begin{figure}
    \includegraphics[width=\columnwidth]{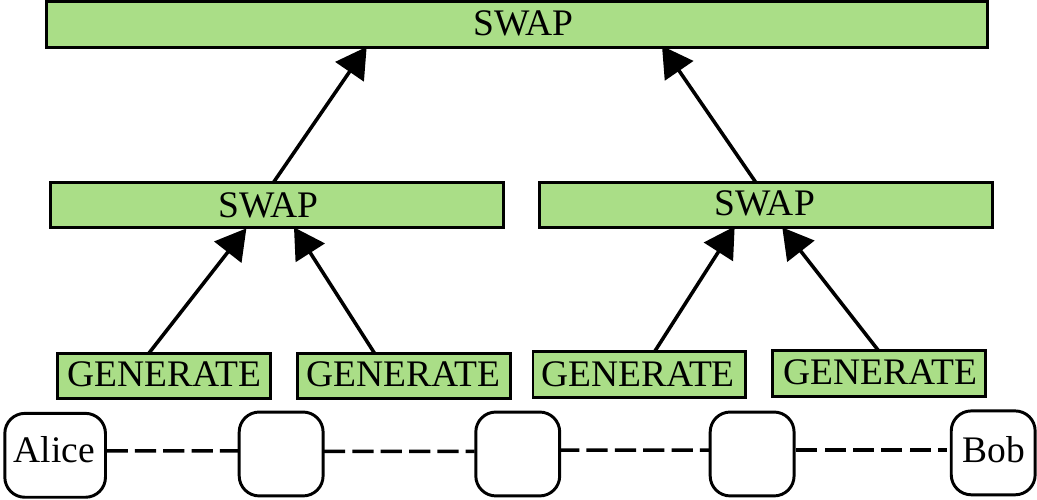}
    \caption{Schematic of a nested repeater protocol on five nodes ($n=2$ nesting levels)
        The figure depicts the protocol for delivering entanglement between remote parties Alice and Bob through three repeater nodes.
		At the start of the protocol, all nodes attempt to generate an elementary link with each of their neighbors in parallel.
		An entanglement swap is performed once the two leftmost links are ready, and similarly for the two rightmost links.
		Once both swaps have succeeded (failure requires regeneration of the involved links), the middle node performs an entanglement swap, which yields entanglement between Alice and Bob.
    \label{fig:schematic-nested-repeater}
	}
\end{figure}

\begin{figure}
	\includegraphics[width=\columnwidth]{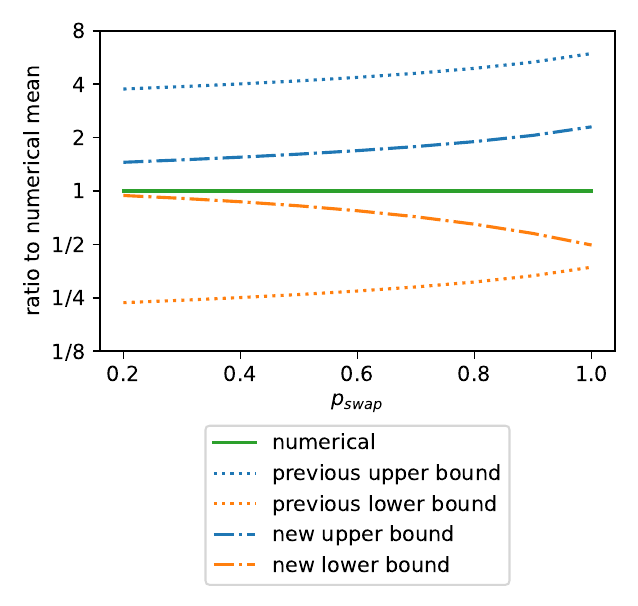}
    \caption{The ratio of different upper and lower bounds on the mean completion time of a nested repeater protocol, as compared to the numerically calculated mean with the deterministic algorithm from \cite{brand2020efficient} ( 17-node repeater chain; $\pgen=0.5$; entanglement generation is performed in discrete attempts).
    The horizontal axis shows the success probability $\pswap$ of entanglement swapping.
    The figure shows bounds known before this work (eq.~\eqref{eq:markov-repeater-mean-bound}) and the tighter bounds from this work in Prop.~\ref{prop:repeater-bounds}\ref{item:repeater-chain-mean-bounds} and (b).
	\label{fig:mean_bound_ratios_plot}
	}
\end{figure}

\begin{figure}[h!]
	\centering
	\includegraphics[width=0.7\columnwidth]{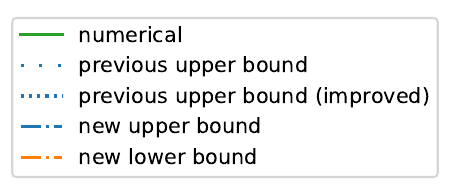}
	\includegraphics[width=\columnwidth]{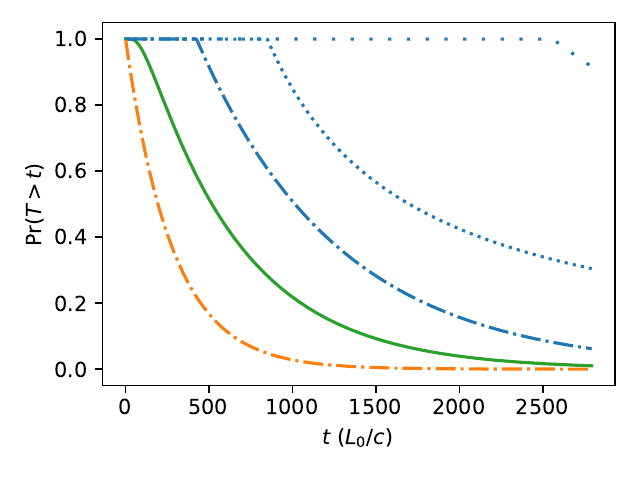}
	\includegraphics[width=\columnwidth]{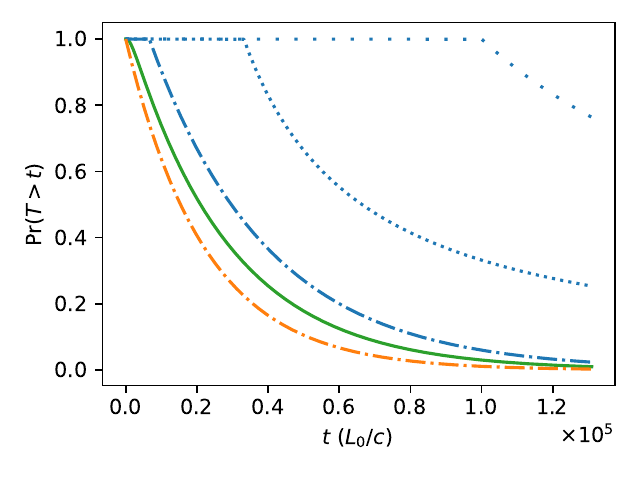}
	\caption{
        Probability distribution of the completion time $T$ of a symmetric nested repeater protocol.
        $T$ is given in units of $L_0/c$, which is the time of a single attempt at elementary-link-generation between neighboring nodes at distance $L_0$ (see also Sec.~\ref{sec:preliminaries-protocols}).
        The figure shows the numerically computed distributions using the deterministic algorithm from \cite{brand2020efficient}, a polynomially-decaying bound known before this work which is derived from Markov's inequality and a bound on the mean completion time (eq.~\eqref{eq:markov-repeater-cdf-bound}), and two improvements on eq.~\eqref{eq:markov-repeater-cdf-bound} we achieve in this work: first, a simple improvement by using Markov's inequality and the improved bound on the mean completion time (Prop.~\ref{prop:repeater-bounds}\ref{item:repeater-chain-mean-bounds}), followed by the exponentially-decaying two-sided tail bounds from Prop.~\ref{prop:repeater-bounds}\ref{item:repeater-chain-tail-bounds}.
        The plot shows results for a repeater chain with 17 nodes ($\pgen=0.1$) where entanglement generation is performed in discrete attempts.
        The swap success probability is $\pswap=0.5$ (top), and $\pswap = 0.2$ (bottom).
	\label{fig:tail_bound_plot}
	}
\end{figure}

By using the notion of stochastic dominance (Def.~\ref{def:stochastic-dominance}), Prop.~\ref{prop:repeater-bounds} can straightforwardly be extended to the asymmetric case, i.e. where the success probabilities for entanglement generation and swapping vary throughout the chain.
This is the case, for example, when the segments are not evenly distributed.
We obtain bounds for an asymmetric repeater chain protocol by noting that its completion time $\Trepeater_n^{\textnormal{asym}}$ is stochastically dominated as 
\begin{equation}
    \Trepeater_n^{\max}
    \stleq
\Trepeater_n^{\textnormal{asym}}
    \stleq
    \Trepeater_n^{\min}
    \label{eq:asym-repeater}
\end{equation}
where $\Trepeater_n^{\max}$ ($\Trepeater_n^{\min}$) is the completion time of the symmetric repeater protocol where all success probabilities are replaced by their maximum (minimum).
We thus arrive at the following corollary (formal proof of eq.~\eqref{eq:asym-repeater} in Appendix~\ref{app:asymmetric-repeater}).

\begin{cor}[Asymmetric nested repeater]
    \label{cor:asymmetric-repeater}
    Consider the variant to the repeater chain protocol from Prop.~\ref{prop:repeater-bounds} where the success probabilities for generation and swapping are different throughout the chain.
    Denote by $\pgen^{\min}$ ($\pgen^{\max}$) and $\pswap^{\min}$ ($\pswap^{\max}$) their minimum (maximum), respectively.
    Then, after replacing $\pgen$ and $\pswap$ by $\pgen^{\min}$ and $\pswap^{\min}$ ($\pgen^{\max}$ and $\pswap^{\max}$), respectively, the upper bounds (lower bounds) to $E[\Trepeater_n]$ from Prop.~\ref{prop:repeater-bounds} still hold, and so do the upper bounds (lower bounds) to $\Pr(\Trepeater_n > t)$.
\end{cor}

We finish this section by noting a stronger two-sided bound on the completion time $T$ of an equally-spaced repeater chain than Prop.~\ref{prop:repeater-bounds}(a-b) in the case of deterministic swapping ($\pswap = 1$).
The number of segments can be any integer $N\geq 2$.
Since we assume that the entanglement swaps take no time (Sec.~\ref{sec:preliminaries-protocols}), the mean completion time for this scenario is \cite{collins2007multiplexed, khatri2019practical}
\begin{equation}
    \label{eq:mean-max-deterministic}
    E[T] = E[\max(\Tgen^{(1)}, \Tgen^{(2)}, \dots, \Tgen^{(N)})]
\end{equation}
where $\Tgen^{(k)}$ is an independent and identically distributed copy of $\Tgen$ and describes the completion time of entanglement generation over the $k^{\textnormal{th}}$ segment.
Eq.~\eqref{eq:mean-max-deterministic} has been shown to equal \cite{bernardes2011rate, praxmeyer2013reposition, shchukin2017waitingPRA, khatri2019practical}
\begin{equation}
    \label{eq:mean-max-deterministic-inf-sum}
E[\max(\Tgen^{(1)} , \dots , \Tgen^{(N)} )]
=
    \sum\limits_{k=1}^N {N \choose k} \frac{(-1)^{k+1}}{1 - (1 - \pgen)^k}
    .
\end{equation}

Unfortunately, since eq.~\eqref{eq:mean-max-deterministic-inf-sum} contains a sum whose length is $N$, it is not obvious how $E[T]$ scales with $N$ or $\pgen$.
To get an idea of the scaling, we could use the fact that for $\pgen \ll 1$, the completion time of entanglement generation (the geometric distribution, which is discrete) is well approximated by a exponential distribution (which is continuous).
Formally, by replacing $\Tgen\rightarrow\Tapprox$, the following approximation to $E[T]$ in eq.~\eqref{eq:mean-max-deterministic-inf-sum} has been derived \cite{shchukin2017waitingPRA, schmidt2019memory-assisted}:
    \begin{equation}
        \label{eq:harmonic-number-approximation}
    E[T] \approx \frac{1}{\pgen} \cdot H_N
    \end{equation}
where
\begin{equation}
    \label{eq:harmonic-number}
        H_N := \sum_{k=1}^{N} \frac{1}{k} = \gamma + \log(N) + O\left(\frac{1}{N}\right)
\end{equation}
		is the $N$-th harmonic number and $\gamma \approx 0.5772$ is the Euler-Mascheroni constant.
    If $\pgen \ll 1$, the approximation works well and shows how $E[T]$ scales in $N$ and $\pgen$.
    For $\pgen$ close to $1$, however, it does not: for example, for $\pgen = 1$ we have $E[T] = 1$ but eq.~\eqref{eq:harmonic-number-approximation} still shows $E[T]$ to scale linearly with $H_N$.

    A fairly tight bound which shows the scaling for all $\pgen$ is obtained in work by Eisenberg~\cite{eisenberg2007expectation}.
    To our knowledge no-one has so far noted it in the context of completion times of quantum network protocols.
    We state it below.

\begin{prop}
    \label{prop:deterministic-swapping}
    \cite{eisenberg2007expectation}
    Suppose that entanglement swapping is deterministic ($\pswap = 1$).
    Let $E[T]$ denote the mean completion time of a repeater chain over $N$ segments.
    Then $E[T]$ is bounded as
    \[
        a \cdot H_N \leq
            E[T]
            \leq
            1 + a \cdot H_N
		\]
        where $H_N$ is the $N$-th harmonic number given in eq.~\eqref{eq:harmonic-number} and
        \[
            a = \mugenupper - 1 = \frac{-1}{\log(1 - \pgen)} = \frac{1}{\pgen} - \frac{1}{2} + O(\pgen)
            .
        \]
\end{prop}

\section{Second application: a quantum switch \label{sec:application-switch}}

Here, we apply our results to a quantum switch.
A quantum switch serves $k$ user nodes.
Each user is connected to the switch by an arm, which produces bipartite entanglement (a link) between switch and user.
As soon as each user has produced a link with the switch, the switch performs a $k$-fuse operation, i.e. a probabilistic operation converting $k$ bipartite links into a single $k$-partite entangled state on the user nodes.

\begin{figure}
    \centering
    \includegraphics[width=0.25\textwidth]{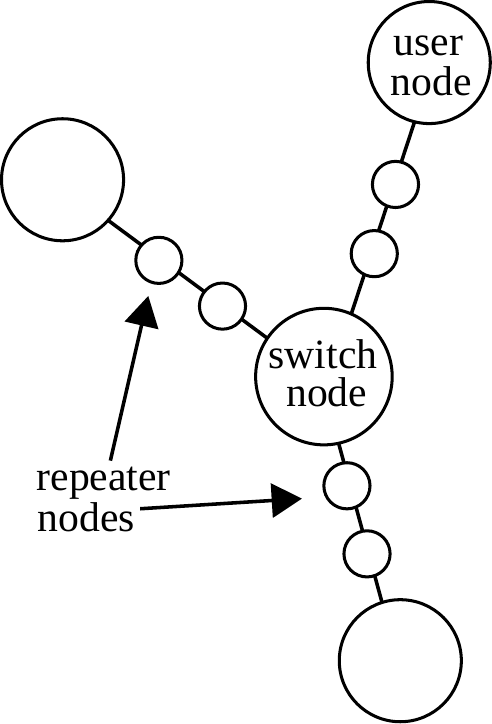}
    \caption{
        \label{fig:switch}
        A quantum switch with $3$ users, each connected to the switch by an identical repeater chain which produces links between user and switch.
        The switch produces $3$-partite entangled states, shared between the users, by performing a probabilistic operation on $3$ links, one with each user node, as soon as these $3$ links are available.
    }
\end{figure}

Vardoyan et al., considered the scenario in which each user produces entanglement continuously with the switch and the switch fuses whenever it can \cite{vardoyan2019stochastic}.
They obtained analytical expressions for the rate at which the switch produces multipartite entanglement in the steady-state regime.
Here, we consider the alternative protocol where the goal is to produce only a single $k$-partite state.
We go beyond the model of Vardoyan et al., by replacing the arms, which connect the switch to the user, by an arbitrary entanglement-distribution network whose completion time is NBU.
An example choice for such a network is the symmetric repeater chain from Sec.\ref{sec:application-repeater}, yielding the network topology as depicted in Fig.~\ref{fig:switch}, 
Our tools allow us to achieve bounds on the completion time of the switch, as described in the following proposition.

\begin{prop}
    \label{prop:switch}
    Consider a $k$-armed quantum switch with fusion success probability $\pfuse$.
    Suppose that the completion times of the different arms are independent and identically distributed according to an NBU random variable $\Tarm$.
    Denote by $\Tswitch$ the time until the switch performs the first successful $k$-fuse attempt.
    Then:
    \begin{enumerate}[(a)]
        \item{$\Tswitch$ is NBU;}
        \item{
                The mean of $\Tswitch$ is bounded as
            \[ 
            E[\Tswitch] \leq \left(k-1 + \frac{1}{k}\right) \cdot \frac{E[\Tarm]}{\pfuse}
            .
            \]
}
        \item{
                $\Tswitch$'s tail decays exponentially fast:
            \[ 
            \Pr(\Tswitch > t) \leq \exp(\pfuse - \frac{\pfuse \cdot t}{(k - 1 + 1/k)\cdot E[\Tarm]})
            .
            \]
}
    \end{enumerate}
\end{prop}

Prop.~\ref{prop:switch}(a) follows directly from Prop.~\ref{prop:swap-bounds-non-iid}(a) (Sec.~\ref{sec:results}).
Prop.~\ref{prop:switch}(b) is a consequence of the expression for the mean completion time in Prop.~\ref{prop:swap-bounds-non-iid}(b) and the upper bound in Prop.~\ref{prop:swap-bounds-non-iid}(g), while Prop.~\ref{prop:switch}(c) is an instantiation of the tail bound of Prop.~\ref{prop:swap-bounds-non-iid}(c) combined with the mean upper bound of Prop.~\ref{prop:switch}(b).

\section{Discussion \label{sec:discussion}}

The distribution of remote entanglement is a key element of many quantum network applications.
In this work, we provided analytical bounds on both the mean and quantiles of entanglement delivery times for a large class of protocols.
We applied these results to a nested quantum repeater chain scheme and to a quantum switch, and obtained bounds which are tighter than present in the literature.

In particular, we considered a frequently-used approximation to the mean entanglement-delivery time in the nested repeater chain scheme, known as the 3-over-2 formula.
This approximation is derived by assuming that the delivery time follows an exponential distribution at each nesting level.
It was not known in general how good this approximation is.
Moreover, finding the exact mean delivery time has been an open problem for more than ten years~\cite{sangouard2011quantum}.
We made a large step towards solving this question by showing that the co-CDF of the delivery time, i.e. the probability that entanglement is delivered after time $t$, is lower bounded by the co-CDF of an exponential distribution, and upper bounded by the co-CDF of an exponential distribution multiplied by a factor which is independent of $t$.
In the limit of small success probabilities of the repeater's components, the bounds coincide.
Second, we show that the 3-over-2 formula is, in essence, an upper bound to the mean delivery time, rendering old analyses building upon this approximation pessimistic.

Regarding future work, note that in many quantum internet scenarios, already-produced entanglement waits for the generation of other entanglement and in the meantime suffers from memory noise.
We leave for future work converting our bounds on the delivery time to bounds on the amount of memory noise, and thus on the quality of the produced state.

In this work we only focused on the first remote entanglement that is delivered.
Some protocols, however, might deliver entanglement while still holding residual entanglement, for example at lower levels in case of the nested repeater chain.
In such a case, it is not optimal to restart the protocol for producing a second entangled pair of qubits, since that would require discarding the residual entanglement.
Hence, another possibility for future work would be to extend our results to protocols which produce multiple entangled pairs without discarding existing entanglement in between.

Our bounds are partially based on a novel connection with reliability theory.
We expect that reliability-theoretic tools will be useful in solving other open problems in quantum networks too.

\section*{Acknowledgements}
The authors would like to thank Kenneth Goodenough, Boxi Li and Filip Rozp\k{e}dek for helpful discussions, and would like to thank Kenneth Goodenough, Boxi Li and Gayane Vardoyan for critical reading of the manuscript.
We thank an anonymous reviewer of this work for noting that Cor.~\ref{cor:asymmetric-repeater} holds.
This work was supported by the QIA project (funded by European Union's Horizon 2020, Grant Agreement No. 820445), the NEASQC project (Horizon 2020, Grant Agreement No. 951821) and by the Netherlands Organization for Scientific Research (NWO/OCW), as part of the Quantum Software Consortium program (project number 024.003.037 / 3368).

\bibliographystyle{IEEEtran}
\bibliography{allbibs}

\appendix
\onecolumngrid
\section{Proofs of main results \label{sec:methods}}
Here, we prove our main results from Sec.~\ref{sec:results}.
We provide proofs in the following order.
First, a proof of Lemma~\ref{lemma:generation-majorization}.
Then, we will prove Prop.~\ref{prop:swap-bounds-non-iid}.
Since Prop.~\ref{prop:swap-bounds} is a special case of Prop.~\ref{prop:swap-bounds-non-iid}, we do not prove it separately.

\subsection{Proof of Lemma~\ref{lemma:generation-majorization}}

Here, we prove the four parts of Lemma~\ref{lemma:generation-majorization}: (i) that $\Tgen$, the completion time of heralded entanglement generation with probability $\pgen$, is stochastically dominated by $\Tupper = 1 + \Texp$, where $\Texp$ is exponentially distributed with parameter $-1 / \log(1 - \pgen)$.
Next, (ii) that the mean of $\Tupper$ equals 
\[
1 - \frac{1}{\log(1-\pgen)} = 
\frac{1}{\pgen} + \frac{1}{2} + O(\pgen)
.
\]
Then, (iii) that $0 \leq E[\Tupper] - E[\Tgen] \leq \frac{1}{2}$ and (iv) that $0 \leq E[\Tupper] / E[\Tgen] \leq 1 + \pgen/2$.
Fifth, (v) that $\Tupper$ is NBU.

Regarding (i), we use the definition of the geometric distribution in eq.~\eqref{eq:geometric-distribution}, from which it follows that the survival function of $\Tgen$ is given by
	\[
        \Pr(\Tgen > t) = (1 - \pgen)^{\lfloor t \rfloor}
		\]
		for all $t \geq 1$, where $\lfloor t \rfloor$ denotes the floor of $t$: $\lfloor t \rfloor = t$ if $t$ is an integer and it equals the largest integer strictly smaller than $t$ otherwise.
    For $0 \leq t < 1$, we have $\Pr(\Tgen > t) = 1 = \Pr(\Tupper > t)$, so the definition of stochastic dominance (Def.~\ref{def:stochastic-dominance}) is trivially satisfied on the interval $t \in [0, 1)$.
    We therefore only need to consider $t \geq 1$.
	Using the notation from Lemma~\ref{lemma:generation-majorization}, we now bound
	\begin{eqnarray*}
		\Pr(\Tgen > t) 
		&=& (1 - \pgen)^{\lfloor t \rfloor} \\
		&\leq& (1 - \pgen)^{t - 1}\\
		&=& \exp\left[ (t - 1) \cdot \log(1 - \pgen)\right]\\
        &\stackrel{*}{=}& \Pr(\Texp > t - 1)\\
		&=& \Pr(1 + \Texp > t)
		,
	\end{eqnarray*}
    where in ${}^*$, we have used the definition of the exponential distribution from eq.~\eqref{eq:exponential-distribution}.
    For proving (ii), we recall that the mean of an exponential distribution with co-CDF $e^{-\lambda t}$ with parameter $\lambda > 0$ is $1/\lambda$, hence the mean of $\Tupper$ is
    \begin{eqnarray*}
        E[\Tupper] &=& E[1 + \Texp] \\
        &=& 1 + E[\Texp] \\
        &=& 1 - \frac{1}{\log(1-\pgen)}\\
        &=& \frac{1}{\pgen} + \frac{1}{2} + O(\pgen)
    \end{eqnarray*}
    where in the last equation, we used the expansion of $1/\log(1+x)$ for $|x|<1$ by Kowalenko \cite{kowalenko2008properties}.
    We show (iii) by computing the derivative of $E[\Tupper] - E[\Tgen]$ as function of $\pgen$, which equals
    \begin{equation}
        \label{eq:derivative}
        \frac{-1}{(1-\pgen) \log^2(1 - \pgen)} + \frac{1}{\pgen^2}
        .
    \end{equation}
    It is not hard to see that eq.~\eqref{eq:derivative} is upper bounded by $0$ for all $\pgen \in (0, 1)$: we start with the well-established inequality\cite{topsoe2007some}
    \[
    \log(x) \geq \frac{x - 1}{\sqrt{x}}
    \]
    for $0 < x \leq 1$, which after the substitution $x \rightarrow 1 - \pgen$ becomes
    \begin{equation}
        \label{eq:derivation-nonpositive-derivative}
    \log(1 - \pgen) \geq \frac{- \pgen}{\sqrt{1 - \pgen}}
    .
    \end{equation}
    Since both sides of eq.~\eqref{eq:derivation-nonpositive-derivative} are negative and the squaring function $x\mapsto x^2$ is monotonically decreasing for $x\leq 0$, squaring both sides requires the inequality sign to flip,
    \[
    \log^2(1 - \pgen) \leq \frac{\pgen^2}{1 - \pgen}
    \]
    and hence
    $(1 - \pgen) \log^2(1 - \pgen) \leq \pgen^2$, implying that the derivative in eq.~\eqref{eq:derivative} is upper bounded by $0$ for all $\pgen\in (0, 1)$.
    Therefore, $E[\Tupper] - E[\Tgen]$ is monotonically decreasing in that regime and achieves its optima at $\pgen \downarrow 0$ and $\pgen \uparrow 1$, which are $\frac{1}{2}$ and $0$, respectively, yielding precisely the bound in (iii).
    For showing (iv), divide each side of $0 \leq E[\Tupper] - E[\Tgen] \leq \frac{1}{2}$ by $E[\Tgen]$ to obtain
    \[
        0 \leq \frac{E[\Tupper]}{E[\Tgen]} - 1 \leq \frac{1}{2E[\Tgen]} = \frac{\pgen}{2}
    \]
    from which (iv) directly follows.
    For proving (v), that $\Tupper = 1 + \Texp$ is an NBU random variable, we consider two cases with respect to the definition of NBU (Def.~\ref{def:nbu}):
	\begin{itemize}
		\item{both $x< 1$ and $y< 1$. Then 
			\[
				\Pr(1 + \Texp > x) = \Pr(1 + \Texp > y) = 1
				\]
				so the definition of NBU trivially holds by the fact that $\Pr(1 + \Texp > x + y)$ cannot exceed 1;}
		\item{at least one of $x$ or $y$ is 1 or larger. Assume without loss of generality that $y \geq 1$. Then note that $\Pr(1 + \Texp > x + y)$ equals
			\begin{eqnarray*}
				&& \Pr(\Texp > x + (y - 1)) \\
				&\leq& \Pr(\Texp > x) \Pr(\Texp > y - 1) \\
				&=& \Pr(\Texp > x) \Pr(1 + \Texp > y)
			\end{eqnarray*}
            where the inequality holds by the fact that $\Texp$ is itself NBU (see Example~\ref{example:exponential-distribution}).
			The proof finishes by noting that $1 + \Texp$ stochastically dominates $\Texp$, i.e. \mbox{$\Pr(1 + \Texp > y) \geq \Pr(\Texp > y)$}.
			}
	\end{itemize}

\subsection{Proof of Proposition~\ref{prop:swap-bounds-non-iid}}

Now, we prove Prop.~\ref{prop:swap-bounds-non-iid}, which automatically proves its special case Prop.~\ref{prop:swap-bounds}.
For our proof, we first give a formal definition of $\Tafter$, following Brand et al. \cite{brand2020efficient}.
The $\restartuntilsuccess$ acts on $n$ quantum states, which first need to have been delivered.
Thus, we define a fresh random variable to refer to the time until the last of $n$ quantum states has been delivered:
\[
	M := \max(T_1, \dots, T_n)
	.
	\]
The restarts of the $\restartuntilsuccess$ protocol, according to a constant success probability $p$, result in the fact that $\Tafter$ can be written as a \textit{geometric sum} of copies of $M$:
	\begin{equation}
		\label{eq:geometric-sum}
	\Tafter = \sum_{k=1}^K M^{(k)}
	\end{equation}
	where $M^{(k)}$ is an i.i.d. copy of $M$ and $K$ is a geometrically distributed random variable with parameter $p$:
	\begin{equation}
		\label{eq:summation-bound}
	\Pr(K = k) = p (1 - p)^{k-1}	
	.
	\end{equation}
Eq.~\eqref{eq:geometric-sum} reflects the fact that the $\restartuntilsuccess$ protocol needs to perform $K$ attempts at success, each of which takes time given by a fresh instance of $M$ (for a more thorough explanation, see \cite{brand2020efficient}).

Now we will prove each of the statements (a-f) from Prop.~\ref{prop:swap-bounds-non-iid}.
For statement (a), we need to show that $\Tafter$ is NBU.
This follows directly from the following two facts: 
\begin{enumerate}[(i)]
	\item{NBU-ness is preserved under the maximum: if $T_1, \dots, T_n$ are NBU random variables, then so is $M$;}
    \item{NBU-ness is preserved under the geometric sum: if $M$ is an NBU random variables, then so is \mbox{$\Tafter = \sum_{k=1}^K M^{(k)}$}.}
\end{enumerate}
We prove item (i) in Sec.~\ref{app:nbu-preserved-under-maximum}, while item (ii) was proven by Brown, see Sec. 3.2 in \cite{brown1990error} \footnote{Let us clarify here that the work by Brown proves that the NBU property is preserved under the geometric sum if $K$ is distributed according to eq.~\eqref{eq:summation-bound}. However, the same paper also proves that if $K$ is shifted by 1, i.e. $\Pr(K = k) = p (1 - p)^k$, then the geometric sum is \textit{always} NWU, irrespective of the summand random variable. However, we will not use the latter case here.}.

Statement (b), $E[\Tafter] = m/p$ with \mbox{$m=E[M]$}, is a simple generalization of \cite[Eq.(2)]{collins2007multiplexed}.
We prove it by applying a well-known fact of randomized sums called Wald's Lemma \cite{wald1947sequential} to eq.~\eqref{eq:geometric-sum}, which is applicable when the length of the sum is independent of the summand.
Applying Wald's lemma results in
\[
    E[\Tafter] = E[M] \cdot E[K]
\]
and hence $E[\Tafter] = m \cdot \frac{1}{p}$.

Statement (c) describes a two-sided bound on the co-CDF of $\Tafter$:
			\[
				\exp\left(\frac{-p\cdot t}{m} \cdot \frac{1}{1 - p}\right)
                \leq
			\Pr(\Tafter > t)
			\leq \exp\left(p - \frac{p \cdot t}{m}\right)
			.
        \]
These bounds follow from the following lemma from Brown, see eq.3.2.4 in \cite{brown1990error}:
\begin{lemma}
	\label{lemma:tail-bounds-brown}
	\cite{brown1990error}
Let $X$ be a real-valued random variable with $\Pr(X < 0) = 0$.
	Define the geometric compound sum of i.i.d. copies of $X$ as $Y := \sum_{k=1}^K X^{(k)}$, where $K$ follows the geometric distribution with success probability $p$ (eq.~\eqref{eq:summation-bound}).
	Moreover, define $Y_0 := \sum_{k=1}^{K_0} X^{(k)}$, where $K_0 = K - 1$.
	Then
	\[
		\Pr(Y > t) \leq \exp(p)\exp\left(- t/ E[Y]\right) \]
		while
	\[
	\Pr(Y > t) \geq \exp\left(- t/ E[Y_0]\right)
	.
		\]	
\end{lemma}
Now interpret $Y\rightarrow \Tafter$ and $X\rightarrow M$ in Lemma~\ref{lemma:tail-bounds-brown}.
The upper bound in statement (c) follows directly from Lemma~\ref{lemma:tail-bounds-brown} by the use of statement (b), which says that $E[\Tafter] = m/p$, while for the lower bound in statement (c) we use 
\begin{eqnarray*}
E[Y_0] 
	&=& E[K_0] \cdot E[X]
	\\&=& E[K_0] \cdot E[M]
    \\&=& \left(\frac{1}{p} - 1\right) \cdot m
	\\&=& (1 - p) \cdot \frac{m}{p}.
\end{eqnarray*}

    Next, (d) states that $\Tafter / E[\Tafter]$ approaches the exponential distribution with mean $1$.
    For proving this statement, we substitute \mbox{$t \rightarrow t \cdot E[\Tafter] = tm/p$} in statement (c).
    The result is a bound on
    \[
        \Pr(\Tafter > t\cdot E[\Tafter])
    =
    \Pr(\Tafter/E[\Tafter] > t)
    \]
    given by
    \[
    \exp\left(-t \cdot \frac{1}{1-p}\right)
    \leq \Pr(\Tafter/E[\Tafter] > t) 
    \leq
    \exp\left(p - t\right) 
    .
    \]
    Letting $p\rightarrow 0$, the bounds on both sides coincide, and thus
    \[
        \lim_{p\rightarrow 0}
    \Pr(\Tafter/E[\Tafter] > t) = \exp\left(-t\right)
    \]
    which is precisely the co-CDF of the exponential distribution with parameter 1.

	For showing the upper bound in statement (e),
	\[
        m \leq \sum_{j=1}^n E[T_j]
		\]
    we use the fact that for all $j=1, \dots, n$, it holds that $T_j \geq 0$.
	The maximum of of nonnegative numbers is upper bounded by its sum, and thus
	\begin{eqnarray*}
		m 
		&=& E[\max(T_1, \dots, T_n)]
		\\&=& \sum_{t_1, \dots, t_n} \Pr(T_1=t_1, \dots, T_n = t_n) \max(t_1, \dots, t_n)
		\\&\leq& \sum_{t_1, \dots, t_n} \Pr(T_1=t_1, \dots, T_n = t_n) \left(t_1+ \dots + t_n\right)
		\\&\stackrel{*}{=}& \sum_{j=1}^n \sum_{t_j}\Pr(T_j=t_j) t_j
		\\&=& E\left[\sum_{j=1}^n T_j\right]
	\end{eqnarray*}
	where for $*$ we made use of the fact that all $T_j$ are independent.
    The proof for the lower bound in statement (e), $\max_{1\leq j \leq n} E[T_j] \leq m$, is similar and relies on the fact that $\max(t_1, \dots, t_n) \geq t_j$ for all $1\leq j\leq n$, where $t_1, \dots, t_n$ are nonnegative numbers.
    Now, we first prove (g) before proving (f).
    Statement (g) states that if all $T_j$ are identically distributed with mean $E[T]$, then
	\[
		1 \leq
			\frac{m}{E[T]} \leq n - 1 + \frac{1}{n}
		\]
        where we recall that $m = E[\max(T_1, \dots, T_n)]$.
        For proving this statement, we need the following lemma from Hu and Lin \cite[Lemma 2.2.]{hu2003characterizations}.

\begin{lemma}
	\label{lemma:nbu-min-bound}
	\cite{hu2003characterizations}
    If $X_1, \dots, X_n$ are independent and identically distributed copies of an NBU random variable $X$ on the domain $[0, \infty)$, then \mbox{$E[\min(X_1, \dots, X_n)] \geq E[X] / n$}.
\end{lemma}
\begin{proof}
	The proof is based on two facts.
	First, note that
	\[
		\Pr(\min(X_1, \dots, X_n) > x) = \prod_{j=1}^n \Pr(X_j > x) = \Pr(X > x)^n.
		\]
	Second, note that if $X$ is NBU, then by repeated application of the definition of NBU (Def.~\ref{def:nbu}), we find that
	\[
		\Pr(X > \sum_{j=1}^n x_j) \leq \prod_{j=1}^n \Pr(X > x_j)
		\]
		for any nonnegative numbers $x_j, 1\leq j \leq n$.
		When choosing all $x_j$ identical, say, to some constant nonnegative number $x$, this reduces to
	\[
		\Pr(X > nx) \leq \Pr(X > x)^n
		.
		\]
	Using these two facts, we can now prove the lemma:
	\begin{eqnarray*}
	E[\min(X_1, \dots, X_n)] 
		&=& \int_{0}^{\infty} \Pr(X > x)^n dx
		\\&\geq& \int_{0}^{\infty} \Pr(X > nx) dx
		\\&=& \int_{0}^{\infty} \Pr(X/n > x) dx
		\\&=& E[X/n]
		\\&=& E[X]/n
	\end{eqnarray*}
	where we have used the fact that for any real-valued random variable $X$ with $\Pr(X < 0) = 0$, the mean can be computed as 
    \begin{equation}
        E[X] = \int_{0}^{\infty} \Pr(X > x) dx
        .
        \label{eq:mean-as-integral}
    \end{equation}
\end{proof}

Statement (g) is proven by noting that for nonnegative numbers $t_1, \dots, t_n$, it holds that $t_j \geq \min(t_1, \dots, t_n)$ for all $j=1, \dots, n$, and therefore
\[
	t_1 + \dots t_n \geq \max(t_1, \dots, t_n) + (n - 1) \cdot \min(t_1, \dots, t_n).
	\]
Translating this to the $T_j$ yields
\begin{eqnarray}
    \nonumber
	E\left[\sum_{j=1}^n T_j\right] 
    &\geq& (n - 1) \cdot E[\min(T_1, \dots, T_n] 
    \\
    &&+ E[\max(T_1, \dots, T_n)]
	.
	\label{eq:max-bound}
\end{eqnarray}
    The left hand side of eq.~\eqref{eq:max-bound} equals $n \cdot E[T]$ by the fact that the $T_j$ are i.i.d., while the right hand side is lower bounded by \mbox{$(n - 1) / n \cdot E[T] + E[\max(T_1, \dots, T_n)]$} by Lemma~\ref{lemma:nbu-min-bound}.
	Reshuffling yields
	\begin{eqnarray*}
	E[\max(T_1, \dots, T_n) 
		&\leq& n \cdot E[T] - \frac{n-1}{n} \cdot E[T] 
		\\&=& \left(n - 1 + \frac{1}{n}\right) \cdot E[T]
		.
	\end{eqnarray*}
	which is what we set out to prove.

    Finally, statement (f) says 
            \begin{eqnarray}
                m = E[\max(T_1, T_2)] \leq \frac{3}{4} \cdot \left(E[T_1] + E[T_2]\right) + \int_{0}^{\infty} \left[ \Pr(T_1 > t) - \Pr(T_2 > t)\right]^2 dt
                \label{eq:statement-(f)}
            \end{eqnarray}
            if $T_1$ and $T_2$ are NBU.
    We prove (f) by first observing that
    \begin{eqnarray}
    2 \Pr(\min(T_1, T_2) > t) 
        & = &
        2 \Pr(T_1 > t) \Pr(T_2 > t) 
        \nonumber
        \\
        &=&
    \Pr(T_1 > t)^2 + \Pr(T_2 > t)^2 - \left[ \Pr(T_1 > t) - \Pr(T_2 > t)\right]^2
        \nonumber
        \\\label{eq:min-inequality}
    \end{eqnarray}
    and therefore, using the fact that the mean can be written as integral over the co-CDF (see also eq.~\eqref{eq:mean-as-integral}), we have
    \begin{eqnarray}
        2 \cdot E[\min(T_1, T_2)]
        &=&
        \int\limits_{0}^{\infty} 2 \Pr(\min(T_1, T_2) > t) dt
        \nonumber
        \\&=&
        \int\limits_{0}^{\infty} \Pr(T_1> t)^2 dt
        +
        \int\limits_{0}^{\infty} \Pr(T_2> t)^2 dt
        -
        \int\limits_{0}^{\infty} \left[\Pr(T_1> t) - \Pr(T_2 > t)\right]^2 dt
        \label{eq:min-integral-(f)}
    \end{eqnarray}
    Now we use the fact that for an NBU random variable $X$, we have $\Pr(X > t)^2 \geq \Pr(X > 2t)$.
    Since $T_1$ and $T_2$ are NBU, we find that
    \[
        \int\limits_{0}^{\infty} \Pr(T_1> t)^2 dt
        \geq
        \int\limits_{0}^{\infty} \Pr(T_1> 2t) dt
        =
        \int\limits_{0}^{\infty} \Pr(T_1 / 2> t) dt
        =
        \frac{1}{2} \cdot E[T_1]
    \]
    and similarly for $T_2$.
    Substituting these inequalities into eq.~\eqref{eq:min-integral-(f)}, we obtain
    \begin{eqnarray*}
        2 \cdot E[\min(T_1, T_2)]
        &\geq&
        \frac{1}{2}\left(E[T_1]
        +
        E[T_2]
        \right)
        -
        \int\limits_{0}^{\infty} \left[\Pr(T_1> t) - \Pr(T_2 > t)\right]^2 dt
    \end{eqnarray*}
    Now invoke $\mean{\max(T_1, T_2)} = \mean{T_1} + \mean{T_2} - \mean{\min(T_1, T_2)}$ to arrive at
    \begin{eqnarray*}
    \mean{\max(T_1, T_2)} 
        &\leq&
        \mean{T_1} + \mean{T_2} - 
        \frac{1}{2}\left( E[T_1] / 2 + E[T_2] / 2 -
        \int\limits_{0}^{\infty} \left[\Pr(T_1> t) - \Pr(T_2 > t)\right]^2 dt
        \right)
        \\&=&
        \frac{3}{4} \left(\mean{T_1} + \mean{T_2}\right)
        +
        \frac{1}{2}
        \int\limits_{0}^{\infty} \left[\Pr(T_1> t) - \Pr(T_2 > t)\right]^2 dt
    \end{eqnarray*}
    which is precisely statement (f).

\subsection{Proof that the NBU property is preserved under the maximum \label{app:nbu-preserved-under-maximum}}

Here, we prove that the NBU property is preserved under the maximum of independent random variables.

\begin{lemma}
    Suppose $X_1, \dots, X_n$ are independent random variables (not necessarily identically distributed).
	If all $X_j$ are NBU random variables, then so is $\max(X_1, \dots, X_n)$.
	\label{lemma:nbu-preserved-under-maximum}
\end{lemma}

We first prove the special case for $n=2$, from which the statement for general $n$ follows.

\begin{lemma}
	Let $A$ and $B$ be independent nonnegative real-valued random variables (not necessarily identically distributed). If both are NBU, then so is $\max(A, B)$.
	\label{lemma:nbu-preserved-under-maximum-two}
\end{lemma}
\begin{proof}
    Let us denote $a_z := \Pr(A > z)$ and $b_z := \Pr(B > z)$ for $z\geq 0$.
    Assume that $A$ and $B$ possess the NBU property (Def.~\ref{def:nbu}), so that
    \begin{equation}
        a_{x+y} \leq a_x a_y \text{ and } b_{x+y} \leq b_x b_y \quad \text{for all $x, y \geq 0$}
        \label{eq:a-and-b-nbu}
        .
    \end{equation}
    We also write $m_z := \Pr(\max(A, B) \geq z)$ and compute
    \begin{eqnarray}
        m_z &=& \Pr(\max(A, B) > z) \nonumber\\
        &=& 1 - \Pr(\max(A, B) \leq z) \nonumber\\
        &=& 1 - \Pr(A \leq z) \Pr(B \leq z) \nonumber\\
        &=& 1 - (1 - a_z)(1 - b_z) 
        \label{eq:mz-easy}
        \\
        &=& a_z + b_z - a_z b_z
        \nonumber
        \\
        &=& a_z + b_z (1 - a_z)
        \label{eq:mz}
        .
    \end{eqnarray}
    We will prove that $\max(A, B)$ is NBU, which in our notation becomes $m_{x+y} \leq m_{x} m_{y}$ for all $x, y \geq 0$.
    To begin, we write out the expressions for both sides, i.e. for $m_{x + y}$ and for $m_x m_y$.
    First, using eq.~\eqref{eq:mz-easy}, we write out 
    \begin{equation}
        \label{eq:mxplusy-before}
    m_{x+y} = 1 - (1 - a_{x+y})(1-b_{x+y})
        .
    \end{equation}
    Since $m_{x+y}$ from eq.~\eqref{eq:mxplusy-before} is monotonically increasing in $a_{x+y}$ and moreover $a_{x+y} \leq a_x a_y$ (eq.~\eqref{eq:a-and-b-nbu}), we obtain
    \begin{equation}
        \label{eq:mxplusy-before-next}
        m_{x+y} \leq 1 - (1 - a_{x}a_{y})(1-b_{x+y})
        .
    \end{equation}
    We use the same insight again, but now for $b_{x+y}$: the right-hand side of eq.~\eqref{eq:mxplusy-before-next} is monotonically increasing in $b_{x+y}$, which combined with the fact that $b_{x+y} \leq b_x b_y$ (eq.~\eqref{eq:a-and-b-nbu}) yields
    \begin{equation}
    m_{x+y} \leq 
        1 - (1 - a_{x}a_{y})(1-b_{x}b_{y})
       = 
        a_x a_y + b_x b_y (1 - a_x a_y)
    .
        \label{eq:mxplusy}
    \end{equation}
    Next, by eq.~\eqref{eq:mz} we have
    \begin{eqnarray}
        m_x m_y 
        &=&
        \left(a_x + b_x (1 - a_x)\right)
        \cdot
        \left(a_y + b_y (1 - a_y)\right)
        \nonumber
        \\
        &=&
        a_x a_y + a_x b_y (1 - a_y) + a_y b_x (1 - a_x) + b_x b_y (1 - a_x) (1 - a_y)
        \label{eq:mxmy}
        .
    \end{eqnarray}
    In order to prove that $m_{x + y} \leq m_x m_y$ we consider three cases.
    \begin{itemize}
    \item{
            \textbf{Case $\boldsymbol{b_x = 0}$}. In this case eq.~\eqref{eq:mxplusy} reduces to $m_{x + y} \leq a_x a_y$ and eq.~\eqref{eq:mxmy} becomes 
            \begin{equation}
            m_x m_y = a_x a_y + a_x b_y (1 - a_y)
                \label{eq:mxmy-special-case}
                .
            \end{equation}
            Since $a_x, a_y, b_x$ and $b_y$ are all cumulative probabilities, they take values in the interval $[0, 1]$, and therefore the second term of eq.~\eqref{eq:mxmy-special-case} is nonnegative, which yields $m_x m_y \geq a_x a_y \geq m_{x + y}$.
        }
    \item{
            \textbf{Case $\boldsymbol{b_y = 0}$}. By the fact that both the right hand side of eq.~\eqref{eq:mxplusy} as well as the expression for $m_x m_y$ (eq.~\eqref{eq:mxmy}) are invariant under exchanging $b_x$ and $b_y$, this case is proven identically to the first case.
        }
    \item{
            \textbf{Case $\boldsymbol{b_x \neq 0}$ and $\boldsymbol{b_y \neq 0}$}.
         Using eq.~\eqref{eq:mxplusy} and eq.~\eqref{eq:mxmy}, we expand
            \begin{eqnarray*}
            \frac{m_{x+y} - m_{x} m_{y}}{b_x b_y}
                &=&
                \frac{a_x a_y}{b_x b_y}
                + \frac{b_x b_y}{b_x b_y} \left(1 - a_x a_y \right)
                - \frac{a_x a_y}{b_x b_y}
                - \frac{a_x b_y}{b_x b_y} \left(1 - a_y\right)
                - \frac{a_y b_x}{b_x b_y} \left(1 - a_x\right)
                - \frac{b_x b_y}{b_x b_y} \left(1 - a_x\right) \cdot \left(1 - a_y\right)
                \\
                &=&
                1 - a_x a_y
                - \frac{a_x}{b_x} \left(1 - a_y\right)
                - \frac{a_y}{b_y} \left(1 - a_x\right)
                - \left(1 - a_x\right) \cdot \left(1 - a_y\right)
            \end{eqnarray*}
            Using the fact that $b_x, b_y \leq 1$, we obtain
            \[
            \frac{m_{x+y} - m_{x} m_{y}}{b_x b_y}
                \leq
                1 - a_x a_y
                - a_x \left(1 - a_y\right)
                - a_y \left(1 - a_x\right)
                - \left(1 - a_x\right) \cdot \left(1 - a_y\right)
                = 0
                .
        \]
            Since $b_x$ and $b_y$ are positive numbers, it follows that $m_{x+y} - m_{x} m_{y} \leq 0$. This concludes our proof.
        }
    \end{itemize}
\end{proof}

Let us now show how Lemma~\ref{lemma:nbu-preserved-under-maximum} follows from Lemma~\ref{lemma:nbu-preserved-under-maximum-two}.
Let $X_1, \dots, X_n$ be $n$ NBU independent random variables, for $n \geq 2$.
We use induction on $n$.
The case $n=2$ is proven in Lemma~\ref{lemma:nbu-preserved-under-maximum-two}.
Now suppose Lemma~\ref{lemma:nbu-preserved-under-maximum} holds for $n=m$ for some $m\geq 2$.
We show that Lemma~\ref{lemma:nbu-preserved-under-maximum-two} also holds for $n=m+1$.
For this, choose $A = \max(X_1, \dots, X_m)$ and $B = X_{m+1}$.
By assumption, $B$ is NBU, and so is $A$ by the induction hypothesis.
Note that 
\begin{eqnarray*}
	\max\left(X_1, \dots, X_m,  X_{m+1}\right) 
	&=& \max\left(\max\left(X_1, \dots, X_m\right),  X_{m+1}\right) 
	\\&=& \max\left(A, B\right)
,
\end{eqnarray*}
so it follows from Lemma~\ref{lemma:nbu-preserved-under-maximum-two} that $\max(X_1, \dots, X_{m+1})$ is also NBU, which concludes the proof of Lemma~\ref{lemma:nbu-preserved-under-maximum}.

\section{Proof of the lower bounds in Proposition~\ref{prop:repeater-bounds} \label{app:repeater-bounds}}

Here, we prove the two lower bounds in Prop.~\ref{prop:repeater-bounds}: first, Prop.~\ref{prop:repeater-bounds}(b), followed by the lower bound on the quantiles from Prop.~\ref{prop:repeater-bounds}(c).

Throughout the appendix, we will use the notation $X^{(1)}, X^{(2)}, \dots$ to denote independent and identically distributed copies of a random variable $X$.
Before proving the bounds on the mean and tail of $\Trepeater_n$, let us formally define it.
Regarding the base case $n=0$, which describes elementary-link generation between adjacent nodes, we use either of two flavors: we either set $\Trepeater_0 = \Tgen$, i.e. $\Trepeater_0$ follows the geometric distribution with parameter $\pgen$, or we set $\Trepeater_0 = \Tapprox$, i.e. $\Trepeater_0$ follows the exponential distribution with parameter $\pgen$.
For each statement about $\Trepeater_n$ in this appendix, either the statement will hold for both flavors, or it will be clear from the context which of the two flavors is used.
Regardless of the choice for $n=0$, we define $\Trepeater_n$ for $n>0$ as
    \begin{equation}
        \label{eq:Tn}
        \Trepeater_{n+1} = \sum_{k=1}^K M_n^{(k)}
    \end{equation}
where $K$ is geometrically distributed with parameter $\pswap$ and $M_n$ is defined as
\begin{equation}
    \label{eq:Mn-repeater}
M_n = \max(\Trepeater_n^{(1)}, \Trepeater_n^{(2)})
.
\end{equation}
    Eq.~\eqref{eq:Tn} was given in \cite{brand2020efficient} and can be found by applying eq.~\eqref{eq:geometric-sum} to each nesting level of the repeater protocol, where $M=M_n$ in eq.~\eqref{eq:geometric-sum} describes the time until the last of two links, each spanning $2^n$ repeater segments, has been delivered.

\subsection{Proof of Proposition~\ref{prop:repeater-bounds}(b)}
\label{app:repeater-lower-bound}
Here, we will prove the lower bound on the mean completion time $\Trepeater_n$ of the nested repeater protocol on $n$ nesting levels.
Informally stated, the insight is that
\[
    \max\left(\sum_{k=1}^{K^{(1)}} X^{(k)}, \sum_{k=1}^{K^{(2)}} X^{(k)}\right)
    \stgeq
\sum_{k=1}^{\max(K^{(1)}, K^{(2)})} X^{(k)}
\quad \textnormal{(informal)}
\]
i.e. considering sums with independent and identically distributed summands, the maximum of two sums stochastically dominates the ``longest'' of the two.
Since the definition of $M_n$ in eq.~\eqref{eq:Mn-repeater} contains the maximum of two such sums, we use this idea to define a new random variable $R_n$ as the ``longest'' of the two sums; by the insight above, $R_n$ is stochastically dominated by $M_n$.
Using Lemma~\ref{lemma:stochastic-dominance-means}, this stochastic domination can be converted to $E[M_n] \geq E[R_n]$, after which the bound on the mean of $\Trepeater_n$ as described in Prop.~\ref{prop:repeater-bounds}(b) follows by noting that $E[\Trepeater_n] = E[M_n] / \pswap$.

We now give the formal proof, which we divide into three steps.
First, we define $R_n$ and compute its mean.
Next, we show that $M_n \stgeq R_n$ for all $n>0$, from which we infer a lower bound on the mean of $\Trepeater_n$ as third step.

For the first step, we define $R_n$:
\begin{eqnarray*}
    R_0 &=& \max(\Trepeater_0^{(1)}, \Trepeater_0^{(2)})
    ,
    \\
    R_{n+1} &=& \sum_{j=1}^N R_{n}^{(j)} \quad\textnormal{for $n\geq 0$}
    .
\end{eqnarray*}
Here, $N = {\max\left(K^{(1)}, K^{(2)}\right)}$ where $K^{(1)}$ and $K^{(2)}$ are both geometrically distributed with parameter $\pswap$.
We emphasize that contrary to $\Trepeater_n$, the random variable $R_n$ does not correspond to the completion time of a protocol.

The mean of $R_n$ is computed using the following two lemmas.

\begin{lemma}
    Let $X^{(1)}$ and $X^{(2)}$ be independent and identically distributed random variables with mean $1/p$ for some $0<p\leq 1$.
    If both $X^{(1)}$ and $X^{(2)}$ follow a geometric distribution, then~\cite{collins2007multiplexed}
    \[
        E[\max(X^{(1)}, X^{(2)})]
 = \frac{3 - 2p}{p(2 - p)}
    \]
    while if they follow an exponential distribution, then
    \[
        E[\max(X^{(1)}, X^{(2)})] = \frac{3}{2p}
.
    \]
    \label{lemma:max-mean-geometric-exponential}
\end{lemma}
\begin{proof}
We start with the case that $X$ follows a geometric distribution.
    Note that $\min(X^{(1)}, X^{(2)})$ is geometrically distributed with parameter $1 - (1 - p)^2$:
\[
    \Pr(\min(X^{(1)}, X^{(2)}) > t) = \Pr(X^{(1)} > t) \Pr(X^{(2)} > t) = (1 - p)^t \cdot (1 - p)^t = (1-p)^{2t} = \left[1 - \left(1 - (1-p)^2\right)\right]^t
\]
for $t = 0, 1, 2, \dots$.
    Combined with the fact that $E[\max(X^{(1)}, X^{(2)})] = E[X^{(1)} + X^{(2)} - \min(X^{(1)}, X^{(2)})]
=
    E[X^{(1)}] + E[X^{(2)}] - E[\min(X^{(1)}, X^{(2)})]
$, we obtain 
    \[
        E[\max(X^{(1)}, X^{(2)})] = \frac{1}{p} + \frac{1}{p} - \frac{1}{1 - (1-p)^2}
    =
 \frac{3 - 2p}{p(2 - p)}
\]
    The case of the exponential distribution is analogous, with $\min(X^{(1)}, X^{(2)})$ following the exponential distribution with parameter $2p$.
\end{proof}

\begin{lemma}
The mean of $R_n$ is
    \begin{equation}
        \label{eq:mean_Rn_bound}
    E[R_n] = \left(\frac{3 - 2 \pswap}{\pswap(2 - \pswap)}\right)^{n} \cdot \nu_0
    \end{equation}
where $\nu_0$ is defined as follows.
If $\Trepeater_0$, which describes elementary-link generation between adjacent nodes, follows the geometric distribution with parameter $\pgen$, then
\begin{equation}
    \label{eq:mean-max-geometric-nu0}
    \nu_0 = E[R_0] = E[\max(\Trepeater_0^{(1)}, \Trepeater_0^{(2)})] = \frac{3 - 2\pgen}{\pgen(2 - \pgen)}
\end{equation}
while if $\Trepeater_0$ follows the exponential distribution with parameter $\pgen$, then
\begin{equation}
    \label{eq:mean-max-exponential-nu0}
    \nu_0 = E[R_0] = E[\max(\Trepeater_0^{(1)}, \Trepeater_0^{(2)})] = \frac{3}{2\pgen}
    .
\end{equation}
\end{lemma}
\begin{proof}
    We use induction on $n$.
    The case $n=0$ is treated in Lemma~\ref{lemma:max-mean-geometric-exponential} where we set $p=\pgen$.
    For the induction case, we note that $N$ and $R_{n}$ are independent, hence we may apply Wald's Lemma \cite{wald1947sequential} to obtain
    \[
        E[R_{n+1}] = E\left[ \sum_{j=1}^N R_{n}^{(j)} \right]
        = E[N] \cdot E[R_n]
        .
    \]
    Since $N = \max(K^{(1)}, K^{(2)})$ and $K$ is geometrically distributed with parameter $\pswap$, we again invoke Lemma~\ref{lemma:max-mean-geometric-exponential} to obtain
\[
    E[N] = E[\max(K^{(1)}, K^{(2)})] = \frac{3 - 2 \pswap}{\pswap(2 - \pswap)}
    .
\]
This finishes the proof.
\end{proof}

As second step, we will show that $M_n$ stochastically dominates $R_n$, for which we need the following two auxiliary lemmas and corollary.

\begin{lemma}
	\label{lemma:stochastic-ordering-max}
	Let $P$ and $Q$ be independent real-valued random variables, and $P'$ and $Q'$ i.i.d. copies of $P$ and $Q$ respectively. Then $P \stgeq Q$ implies $\max(P, P') \stgeq \max(Q, Q')$.
\end{lemma}
\begin{proof}
	By definition of $P \stgeq Q$, we have, for all real numbers $z$, that
	$
	\Pr(P > z) \geq \Pr(Q > z)
	$
    and therefore
	$
    \Pr(P \leq z) \leq \Pr(Q \leq z).
	$
	Consequently,
	\[
	\Pr(\max(P, P') > z) = 1 - \Pr(\max(P, P') \leq z) = 1 - \Pr(P \leq z)^2
	\geq
	1 - \Pr(Q \leq z)^2 = \Pr(\max(Q, Q') > z)
	\]
    for all real numbers $z$, so $\max(P, P') \stgeq \max(Q, Q')$.
\end{proof}

\begin{lemma}
    \label{lemma:insight}
    Let $P$ and $Q$ be independent, real-valued random variables with identical domain. Then $\max(P, Q) \stgeq Q$.
\end{lemma}
\begin{proof}
    For any real number $z$, we have
    \[
    \Pr(\max(P, Q) > z)
    = 1 - \Pr(\max(P, Q) \leq z)
    = 1 - \Pr(P \leq z)\Pr(Q \leq z)
    \stackrel{*}{\geq} 1 - \Pr(Q \leq z)
    = \Pr(Q > z)
    \]
    where the inequality * holds because $\Pr(P < z) \leq 1$.
\end{proof}

\begin{cor}
    \label{lemma:sum-length-stgeq}
    Let $A^{(1)}, A^{(2)}, A^{(3)}$ and $A^{(4)}$ be independent and identically distributed random variables with domain $\{1, 2, 3, \dots\}$.
    Furthermore, let $X, Y$ and $Z$ be independent and identically distributed random variables with domain $[0, \infty)$.
    Then
\begin{equation}
    \label{eq:sum-length-stgeq}
        \max\left(
        \sum_{a=1}^{A^{(1)}} X^{(a)},
        \sum_{b=1}^{A^{(2)}} Y^{(b)}
        \right)
        \stgeq
        \sum_{a=1}^{\max\left(A^{(3)}, A^{(4)}\right)} Z^{(a)}
        .
\end{equation}
\end{cor}
\begin{proof}
    We note that random sums occur on both sides of eq.~\eqref{eq:sum-length-stgeq}, that is, sums whose number of terms is a random variable.
    We expand both sides of the inequality from the lemma as a weighted sum over instantiations of this random variable.
    For the left-hand-side, we obtain
    \[
        \Pr( 
        \max\left(
        \sum_{a=1}^{A^{(1)}} X^{(a)},
        \sum_{b=1}^{A^{(2)}} Y^{(b)}
        \right)
        > y)
        =
        \sum_{i=1}^{\infty} \sum_{j=1}^{\infty} \Pr(A^{(1)} = i) \cdot \Pr(A^{(2)}=j) \cdot C^y_{ij}
    \]
    for $y\geq 0$, where we have defined
    \[
        C^y_{ij} := \Pr(\max\left(\sum_{a=1}^i X^{(a)}, \sum_{b=1}^j Y^{(b)}\right) > y)
    \]
    and for the right-hand-side we get
    \[
        \Pr(
        \sum_{a=1}^{\max\left(A^{(3)}, A^{(4)}\right)} Z^{(a)}
        > y
        )
        =
        \sum_{i=1}^{\infty} \sum_{j=1}^{\infty} \Pr(A^{(3)} = i) \cdot \Pr(A^{(4)}=j) \cdot D^y_{ij}
    \]
    with
    \[
        D^y_{ij} :=
    \Pr(\sum_{a=1}^{\max(i,j)} Z^{(a)} > y)
.
    \]
    Given fixed $i$ and $j$, we define random variables $P$ and $Q$ as follows:
    \begin{itemize}
        \item if $\max(i, j) = i>j$, then define $P= \sum_{b=1}^{j} Y^{(b)}$ and $Q= \sum_{a=1}^{i} X^{(a)}$;
        \item if $\max(i, j) = j$, then define $P= \sum_{a=1}^{i} X^{(a)}$ and $Q= \sum_{b=1}^{j} Y^{(b)}$;
    \end{itemize}
    In both cases, application of Lemma~\ref{lemma:insight} that $\max(P, Q) \stgeq Q$ yields $C^y_{ij} \geq \Pr(\sum_{a=1}^{\max(i,j)} Y^{(a)} > y)$.
    Since $Y$ and $Z$ are i.i.d., we obtain $C^y_{ij} \geq D^y_{ij}$ for all $y\geq 0$ and for all $i, j$.
    This concludes the proof.
\end{proof}

Now we have the tools to show that $M_n$ stochastically dominates $R_n$, as described in the following lemma.

\begin{lemma}
    \label{lemma:repeater-mean-lower-bound}
    For all $n\geq 0$, we have
    \[
       M_n \stgeq R_{n}
    \]
    where $M_n = \max(\Trepeater_n^{(1)}, \Trepeater_n^{(2)})$ as defined in eq.~\eqref{eq:Mn-repeater}.
\end{lemma}
\begin{proof}
    We use induction on $n$.
    The base case $n=0$ is an equality by definition of $R_0$.
    Now assume the statement from the lemma holds for $n=m$.
    We will show it also holds for $n=m+1$.
    First, we expand the definition of $\Trepeater_{m+1}$:
    \[
        \Trepeater_{m+1} = \sum_{k=1}^K \max(\Trepeater_m^{(1)}, \Trepeater_m^{(2)})
    \]
    Now apply the induction hypothesis:
    \[
        \Trepeater_{m+1} \stgeq \sum_{k=1}^K R_m^{(k)}
        .
        \]
        Using Lemma~\ref{lemma:stochastic-ordering-max} we obtain
        \[
            \max(\Trepeater_{m+1}^{(1)}, \Trepeater_{m+1}^{(2)}) \stgeq \max\left( \sum_{j=1}^{K^{(1)}} R_m^{(i)}, \sum_{j=1}^{K^{(2)}} R_m^{(j)}\right)
        .
        \]
    Applying Corollary~\ref{lemma:sum-length-stgeq} to the previous equation yields
    \[
        \max(\Trepeater_{m+1}^{(1)}, \Trepeater_{m+1}^{(2)}) \stgeq \sum_{k=1}^{\max(K^{(1)},K^{(2)})} R_m^{(k)}
        .
        \]
        The left-hand side of the previous equation equals $M_{m+1}$ by definition, while its right-hand side is $R_{m+1}$, again by definition.
        This concludes the proof.
\end{proof}

The third step is to derive the lower bound on the mean delivery time from Prop.~\ref{prop:repeater-bounds}.
This follows directly from Lemma~\ref{lemma:repeater-mean-lower-bound}, as expressed in the following corollary.

\begin{cor}
    \textnormal{(Lower bound from Prop.~\ref{prop:repeater-bounds})}
    For $n>0$, it holds that
    \[
    E[\Trepeater_n] \geq \frac{1}{\pswap} \cdot \left(\frac{3 - 2 \pswap}{\pswap(2 - \pswap)}\right)^{n-1} \cdot \nu_0
    \]
    where $\nu_0$ is given in eq.~\eqref{eq:mean-max-geometric-nu0} or eq.~\eqref{eq:mean-max-exponential-nu0}, depending on whether elementary-link generation is modelled following a geometric or exponential distribution, respectively.
\end{cor}
\begin{proof}
    By Wald's Lemma \cite{wald1947sequential} and the fact that $K$ and $M_{n-1}$ are independent, it follows from the definition of $\Trepeater_n$ for $n>0$ that $E[\Trepeater_n] = E[K] \cdot E[M_{n-1}] = \frac{1}{\pswap} \cdot E[M_{n-1}]$.
    A lower bound on $E[M_n]$ follows from Lemma~\ref{lemma:stochastic-dominance-means} and Lemma~\ref{lemma:repeater-mean-lower-bound}, resulting into
    \[
        E[\Trepeater_n] = \frac{1}{\pswap} \cdot E[M_{n-1}]
        \geq \frac{1}{\pswap} \cdot E[R_{n-1}]
        .
    \]
    The proof finishes by substituting $E[R_{n-1}]$ by the right-hand side of eq.~\eqref{eq:mean_Rn_bound}.
\end{proof}

\subsection{Proof of lower bound in Proposition~\ref{prop:repeater-bounds}(b)}
Here, we provide the expression for $\mlower$ in Prop.~\ref{prop:repeater-bounds}(c), which is a lower bound to the mean of the delivery time after both input links are ready, but before the entanglement swap.
Formally, $\mlower$ is a lower bound to the mean of $M_{n-1}$ from eq.~\eqref{eq:Mn-repeater}.
Such a bound follows directly from Lemma~\ref{lemma:repeater-mean-lower-bound} by the fact that $X \stgeq Y$ implies $E[X] \geq E[Y]$ (see Lemma~\ref{lemma:stochastic-dominance-means}):
\[
    \mlower = E\left[R_{n-1}\right]
\]
and $E[R_{n-1}]$ is given in eq.~\eqref{eq:mean_Rn_bound}.

\section{Proof of Cor.~\ref{cor:asymmetric-repeater} for asymmetric nested repeater chains}
\label{app:asymmetric-repeater}

Here, we sketch the proof of the following proposition (see also eq.~\eqref{eq:asym-repeater}), from which Cor.~\ref{cor:asymmetric-repeater} immediately follows.

\begin{prop}
\label{prop:asym-repeater}
Denote by $\Trepeater_n^{\textnormal{asym}}$ the completion time of a nested repeater chain with $n$ levels (see Sec.~\ref{sec:application-repeater}) where the success probabilities for entanglement generation and entanglement swapping are not constant throughout the chain.
By $\Trepeater_n^{\max}$ ($\Trepeater_n^{\min}$) denote the completion time of the symmetric repeater protocol where all success probabilities are replaced by their maximum (minimum), denoted as $\pgen^{\max}$ and $\pswap^{\max}$ ($\pgen^{\min}$ and $\pswap^{\min}$).
Then
\[
    \Trepeater_n^{\max}
    \stleq
\Trepeater_n^{\textnormal{asym}}
    \stleq
    \Trepeater_n^{\min}
\]
where $\stleq$ denotes stochastic domination (Def.~\ref{def:stochastic-dominance}).
\end{prop}

For proving Prop:~\ref{prop:asym-repeater}, we need the following lemma.

\begin{lemma}[Stochastic domination preserved under maxima and geometric sums]
\label{lemma:stdom-preserved}
Let $A_j, B_j$ ($1\leq j\leq n$) be independent random variables, taking values in the nonnegative real numbers.
Furthermore, let $K$ and $M$ be independent random variables, geometrically distributed with parameters $p_K$ and $p_M$, respectively.
Then:
\begin{enumerate}[(1)]
\item If $p_M \leq p_K$, then $K \stleq M$;
\item If for all $j$, $A_j \stleq B_j$, then $\max(A_1, A_2, \dots, A_n) \stleq \max(B_1, B_2, \dots, B_n)$
\item If $K \stleq M$ and $A_1 \stleq B_1$, then $\sum_{k=1}^{K} A_1^{(k)} \stleq \sum_{m=1}^{M} B_1^{(m)}$.
\end{enumerate}
\end{lemma}
\begin{proof}
Statement (1) is proven as $\Pr(K > t) = (1 - p_K)^t \geq (1 - p_M)^t = \Pr(M > t)$ for any $t\in \{1, 2, \dots\}$.
For (2), we write 
\begin{eqnarray*}
\Pr(\max(A_1, A_2, \dots, A_n) > t) 
&=&
1 - \Pr(\max(A_1, A_2, \dots, A_n) \leq t) 
\\ &=&
1 - \Pr(A_1 \leq t) \cdot \dots \cdot \Pr(A_n \leq t)
\\
&\leq&
1 - \Pr(B_1 \leq t) \cdot \dots \cdot \Pr(B_n \leq t)
\qquad\textnormal{because $A_j \stleq B_j$ for all $j$}
\\ &=&
\Pr(\max(B_1, B_2, \dots, B_n) > t)
\end{eqnarray*}
for any $t\geq 0$.
Statement (3) was proven as Lemma 2(e) in Appendix B of \cite{brand2020efficient}.
\end{proof}

With Lemma~\ref{lemma:stdom-preserved}, Prop.~\ref{prop:asym-repeater} is most easily proven by induction over the number of nesting levels, following the definition of $\Trepeater_n$ as given at the start of App.~\ref{app:repeater-bounds}.
For elementary links, note that the elementary-link delivery time of $\Trepeater^{\textnormal{asym}}$ stochastically dominates $\Trepeater^{\max}$ and is stochastically dominated by $\Trepeater^{\min}$, by Lemma~\ref{lemma:stdom-preserved}(i).
For the induction step, first both quantum states which are inputted to the entanglement swap need to be prepared.
By Lemma~\ref{lemma:stdom-preserved}(ii), the time this takes in the asymmetric case again stochastically dominates $\Trepeater^{\max}$ and is stochastically dominated by $\Trepeater^{\min}$.
The induction case is finished by noting that a similar ordering holds for the completion time after the entanglement swap, by Lemma~\ref{lemma:stdom-preserved}(iii).

\end{document}